\documentclass[11pt,a4paper]{article}
\usepackage[margin=1.05in]{geometry}
\usepackage{amsmath,amsthm,amssymb,amsfonts,mathtools,bm}
\usepackage[round,authoryear]{natbib}
\usepackage{iftex}
\ifptex
  \PassOptionsToPackage{dvipdfmx}{graphicx}
  \PassOptionsToPackage{dvipdfmx}{xcolor}
  \PassOptionsToPackage{dvipdfmx}{hyperref}
\fi
\usepackage{graphicx}

\usepackage{booktabs}
\usepackage{subcaption}
\usepackage{tikz}
\usetikzlibrary{arrows.meta,calc,positioning}
\usepackage{threeparttable}
\usepackage{tabularx}
\usepackage{float}
\usepackage[colorlinks=true,linkcolor=blue,citecolor=blue,urlcolor=blue]{hyperref}
\usepackage{setspace}
\usepackage{bbm}
\theoremstyle{plain}
\newtheorem{theorem}{Theorem}[section]

\newtheorem{prop}[theorem]{Proposition}
\newtheorem{cor}[theorem]{Corollary}

\theoremstyle{definition}

\newtheorem{assumption}{Assumption}[section]
\newtheorem{example}{Example}[section]

\theoremstyle{remark}
\newtheorem{remark}{Remark}[section]

\newcommand{\1}{\mathbbm{1}}
\numberwithin{equation}{section}

\title{Identification and Estimation of Staggered Difference-in-Differences with Network Spillovers}
\author{Hayato Tagawa\thanks{Graduate School of Economics, University of Tokyo, Bunkyo-ku, Tokyo 113-0033, Japan; Email: hayato-tagawa@g.ecc.u-tokyo.ac.jp}}
\date{August 2026}

\begin{document}

\maketitle

\begin{abstract}
  This paper studies staggered difference-in-differences designs in which treated and comparison units can be exposed to other units' adoption.
  Using an exposure mapping specified by the researcher, we define realized and counterfactual exposure states and decompose the total effect into a switching effect of own adoption and a spillover effect under no own adoption.
  Conditional parallel trends identify the total and switching effects by matching, respectively, the counterfactual and realized exposure states of the treated cohort to the same states among units that never adopt.
  The difference between these effects identifies the spillover effect for the treated cohort.
  We construct outcome regression estimators for these effects and doubly robust estimators for the total and switching effects, and establish their asymptotic normality under spatial dependence.
  Monte Carlo simulations assess finite-sample bias, RMSE, pointwise coverage, and robustness to outcome regression misspecification.
  In an application to Community Health Centers, the estimated total and spillover effects on mortality are negative after treatment.
  We also find that the spillover estimates for never-treated counties are negative at some event times.
\end{abstract}



\newpage
\section{Introduction}

Staggered difference-in-differences (DID) designs study policies that different units adopt at different dates.
The conventional potential outcomes framework for DID assumes the stable unit treatment value assumption (SUTVA).
Its no-interference component requires a unit's outcome to be unaffected by other units' adoption.
Under a parallel trends assumption, outcome changes for units that have not yet adopted or never adopt identify the counterfactual trend for a treated cohort.
The requirement of no interference may fail when a policy operates across locations or network links.
When such spillovers are present, comparison units can be exposed and treated outcomes can combine own adoption with exposure to adoption elsewhere.
A conventional DID comparison therefore need not recover either the effect of own adoption or the effect of the policy as a whole.

The rollout of Community Health Centers studied below illustrates this form of interference.
A county without a center may benefit from a center in a nearby county through travel and referrals, while a county with its own center may also benefit from centers opened elsewhere.

\citet{goodman2021difference} and \citet{sun2021estimating} characterize the comparisons used by regressions with two-way fixed effects and event study regressions when treatment effects are heterogeneous.
\citet{callaway2021difference} define effects by cohort and time, while \citet{borusyak2024revisiting} construct an estimator from untreated outcome paths.
\citet{delgado2015difference} and \citet{clarke2017estimating} study how spillovers affect DID comparisons in spatial settings.
\citet{fiorini2025simple} use comparison units without spillovers, \citet{XU2026106304} defines direct effects at fixed exposure levels, and \citet{butts2024difference} distinguishes switching and total effects.

The policy comparison studied here is defined by realized and counterfactual exposure levels.
Following \citet{10.1214/16-AOAS1005}, the exposure mapping reduces the adoption profile to states relevant for outcomes.
The mapping may return a vector and may retain current and lagged network exposure.
The researcher specifies the exposure mapping as part of the research design.
The realized adoption profile induces the exposure level $H^1_{it}$.
For each treated cohort, we specify a counterfactual adoption scenario under which that cohort remains untreated; this scenario induces $H^0_{it}$ for every unit.
The resulting pair $(H^1_{it},H^0_{it})$ defines the causal comparison.

The total effect admits a decomposition with two additional objects.
The switching effect changes a cohort's own adoption while holding its realized exposure fixed.
The spillover effect changes exposure while the cohort remains untreated.
The total effect equals the sum of these two effects by definition.

We identify the total effect directly at the counterfactual exposure level induced by this scenario.
Conditional parallel trends equates the untreated outcome trend of cohort $g$ at its counterfactual exposure state with the corresponding trend among units that never adopt and are observed at the same exposure state.
The switching comparison instead matches the realized exposure state.
The difference between the identified total and switching effects then identifies the spillover effect for the treated cohort under no own adoption.
The comparison at $H^0_{it}$ also identifies the spillover effect among units that never adopt relative to the same counterfactual exposure for that cohort.
An alternative identification result uses parallel trends across exposure states among units that never adopt and transportability of untreated exposure effects to identify the spillover effect for cohort $g$ under no own adoption.
Adding the switching effect then identifies the corresponding total effect.

We construct outcome regression estimators for these effects.
For a cohort that adopts, the spillover estimator is the difference between the corresponding total and switching estimators.
We also construct doubly robust estimators for the total and switching effects.
The population functional for each doubly robust estimator equals its target when either its outcome regression or the odds constructed from its propensity score are correctly specified.
We derive the joint asymptotic distribution of the effect estimators under spatial dependence and construct pointwise confidence intervals using a spatial HAC covariance estimator.

Monte Carlo simulations evaluate bias, RMSE, and pointwise coverage under correct specification and under misspecification of the switching outcome regression.

We study the rollout of Community Health Centers.
The total effect comparison uses never-treated counties that remain outside the 75-mile exposure radius at the target date.
Both outcome regressions adjust for the baseline characteristics used by \citet{bailey2015war}.
The estimated total effect on mortality is negative at event times $l = 0$ through $l = 4$.
The pointwise confidence intervals for the total effect exclude zero from $l = 1$ onward, and those for the spillover effect for treated counties exclude zero at every event time.

\subsection*{Related literature}

The standard DID literature provides several approaches to estimation when adoption dates and treatment effects vary.
In addition to the methods discussed above, \citet{gardner2025two} propose estimation in two stages, and \citet{ROTH20232218} review recent developments in this literature.
We retain the cohort and event-time indexing of \citet{callaway2021difference} and \citet{sun2021estimating} and allow outcomes to depend on other units' adoption.
Without interference, the switching and total effects reduce to the usual group--time average treatment effect for a cohort at a given time.

\citet{10.1214/16-AOAS1005} use exposure mappings to summarize the assignment vector under interference, and \citet{savje2024causal} studies the consequences of misspecifying those mappings.
\citet{delgado2015difference} and \citet{clarke2017estimating} study DID with spatial spillovers, while \citet{berg2021spillover} discuss spillovers in empirical corporate finance.
\citet{leung2022causal} studies approximate neighborhood interference, and \citet{huber2021framework} separate individual treatment effects from spillover effects.
We identify the spillover effect among units that never adopt relative to the counterfactual exposure specified for each cohort without requiring those units to be unexposed.
For a treated cohort, the alternative identification result combines parallel trends across exposure states with transportability.

The closest DID approaches differ in their target effects and identifying comparisons.
\citet{butts2024difference} distinguishes switching and total effects in a spatial design with two periods and gives an extension to event studies.
We define the decomposition for each cohort and event time, identify the total effect at the counterfactual exposure level induced by the scenario specified by the researcher, and obtain the spillover component as the difference between the total and switching effects.
\citet{fiorini2025simple} target the ATT without interference by ruling out spillovers for treated units after adoption and requiring comparison units to be free of spillovers.
Our setting allows both treated units and units that never adopt to remain exposed.
\citet{XU2026106304} studies direct effects at fixed exposure levels and derives DR estimators with spatially robust inference using finite-dimensional estimating equations.
Our inference analysis uses a related finite-dimensional formulation for the conditional means and propensity scores.
\citet{10.1257/pandp.20261108} studies interference in panels with multiple periods after treatment and potentially staggered adoption.
Our effects average over the realized and counterfactual exposure states of the treated cohort.

\citet{lasso2026spillover} decomposes the canonical DID estimator with two periods into comparisons across exposure states and proposes spill-imputation using outcomes observed before treatment and a pure control group.
We identify the total and switching effects by comparing trends at the same exposure state in a staggered design.
\citet{sun2025differenceindifferencesnetworkinterference} studies ATT parameters for own treatment and spillovers under network interference while adjusting for network confounders.
Our target is the total effect defined by the realized and counterfactual exposure levels and its decomposition into switching and spillover effects.

Section~\ref{sec:setup} defines the effects, Sections~\ref{sec:identification}--\ref{sec:inference} give identification, estimation, and inference, Section~\ref{sec:simulation} reports the Monte Carlo results, and Section~\ref{sec:application} presents the empirical application.

\section{Setup}\label{sec:setup}

\subsection{Framework}\label{sec:framework}

Consider a panel of $N$ units indexed by $i=1,\dots,N$ over periods $t=1,\dots,T$.
Let $D_{it}\in\{0,1\}$ denote own treatment status.
We restrict attention to panels in which all units are untreated in the first period, so $D_{i1}=0$.
We impose absorbing treatment, $D_{it}\le D_{i,t+1}$ for all $i$ and $t<T$.
Let $\mathcal G\equiv\{2,\dots,T\}\cup\{\infty\}$.
Define the adoption time $G_i \equiv \min\{t:D_{it}=1\}\in\mathcal G$, with $G_i=\infty$ for units that never adopt.
Let $\mathbf G=(G_1,\dots,G_N)$ denote the vector of adoption times, and let $Y_{it}(\mathbf G)$ be the potential outcome for unit $i$ at time $t$ under adoption profile $\mathbf G$.
Before imposing an exposure mapping, $Y_{it}(\mathbf G)$ may depend on the full adoption profile.

Let $\mathcal D_N$ collect the locations, distances, network links, and other predetermined information used to construct exposure.
Let $\mathcal C_N$ be the sigma-field generated by $\mathcal D_N$ and by the predetermined covariates and analysis weights introduced below.
Throughout, expectations and probabilities are evaluated under the law that selects a unit uniformly from $\{1,\ldots,N\}$ and then uses its conditional law given $\mathcal C_N$.
Equivalently, they integrate the average of the conditional laws for the $N$ units; conditioning on $\mathcal C_N$ is suppressed from the notation.
Let $\mathcal H_N$ denote the resulting set of exposure states, with $0\in\mathcal H_N$ denoting no exposure.
For $\mathbf g\in\mathcal G^N$, let $h_{it,N}(\mathbf g_{-i};\mathcal D_N)\in\mathcal H_N$ denote an exposure mapping specified by the researcher that summarizes the features of other units' adoption paths that affect unit $i$.
For the observed adoption profile, write $\mathbf G^1\equiv\mathbf G$ and define the realized exposure state by $H^1_{it} \equiv h_{it,N}(\mathbf G^1_{-i};\mathcal D_N) \in \mathcal H_N$.
The mapping restricts the interference structure by requiring own adoption and exposure to be sufficient for the potential outcome.
\begin{assumption}[Exposure mapping]\label{assumption:exposure_mapping}
  For any unit $i$, time $t$, and adoption profiles $\mathbf g,\mathbf g'\in\mathcal G^N$, if $g_i=g_i'$ and $h_{it,N}(\mathbf g_{-i};\mathcal D_N)=h_{it,N}(\mathbf g'_{-i};\mathcal D_N)$, then $Y_{it}(\mathbf g)=Y_{it}(\mathbf g')$.
\end{assumption}

\begin{assumption}[Consistency]\label{assumption:consistency}
  For the realized adoption profile $\mathbf G$, the observed outcome satisfies $Y_{it}=Y_{it}(\mathbf G)$.
\end{assumption}

Under Assumption~\ref{assumption:exposure_mapping}, let $Y_{it}(g,h)$ denote the common value of $Y_{it}(\mathbf g)$ among profiles with $g_i=g$ and $h_{it,N}(\mathbf g_{-i};\mathcal D_N)=h$.
The states $(\infty,0)$, $(\infty,h)$, and $(g,h)$ denote pure control, untreated exposure, and treatment with exposure, respectively.
All causal effects below are defined relative to the chosen mapping.

The mapping is part of the empirical design and is specified before outcomes are analyzed.
The exposure state may be a scalar or a vector.
For example, a vector may contain the number of adopted neighbors at dates $t$ and $t-1$, or current network exposure together with duration or cumulative exposure.
The example based on a line network in Figure~\ref{fig:line_network_rollouts} uses a binary state for current exposure.
This formulation follows the literature on exposure mappings by treating the mapping as the restriction that connects the full assignment profile to potential outcomes \citep{10.1214/16-AOAS1005,VAZQUEZBARE2023105237}.
The analysis conditions on the mapping specified by the researcher.
Mapping misspecification changes the causal states being compared and may violate Assumption~\ref{assumption:exposure_mapping} \citep{savje2024causal}.

The mapping may also represent a decaying exposure history.
For example, a raw exposure index may be
\begin{align*}
  \widetilde H_{it,N}^{decay} = \sum_{j \neq i} w_{ij}\1\{G_j \le t\}\exp\{-\chi(t - G_j)\},
\end{align*}
where $\chi\ge0$ is the decay rate.
This index can be coarsened into states such as zero, low, medium, and high exposure.
The coarsening disregards exposure differences that are not substantively relevant and creates repeated states that can be compared in the data.
Following \citet{10.1214/16-AOAS1005}, \citet{VAZQUEZBARE2023105237}, and \citet{butts2024difference}, we assume $\sup_N |\mathcal H_N| < \infty$.
The underlying distance, dose, or duration index may remain continuous.
Any time-varying network in $\mathcal D_N$ must be predetermined, and $H^1_{it}$ may depend only on adoption histories through $t$.

\subsection{Policy effects}\label{sec:policy-effects}

Fix a cohort $g$, a nonnegative event time $l$, and the corresponding calendar time $t=g+l$.
Let $H^0_{it}\in\mathcal H_N$ denote the exposure level induced for unit $i$ by the counterfactual adoption scenario specified by the researcher for cohort $g$, under which members of the cohort remain untreated.
The same scenario induces $H^0_{it}$ for units inside and outside cohort $g$.
The scenario may, for example, withhold cohort $g$ while leaving the adoption times of other cohorts unchanged.
Figure~\ref{fig:line_network_rollouts} contrasts the realized adoption profile with the profile generated by withholding cohort $g$.

\begin{figure}[h]
  \centering
  \begingroup
  \def\LineNetworkEarlierExposureLabel{exposure from earlier adopter}
  \begingroup
\def\LineNetworkRolloutOnly{}
\def\LineNetworkSuppressLegend{}

\begin{minipage}{\textwidth}
  \centering
  {\small A. Observed adoption profile\par}
  \vspace{0.5em}
  \input{tikz/fig_line_network_spillover_concept.tex}
\end{minipage}

\vspace{1em}

\begin{minipage}{\textwidth}
  \centering
  {\small B. Counterfactual adoption profile with cohort $g$ withheld\par}
  \vspace{0.5em}
  \input{tikz/fig_line_network_cohort_deletion.tex}
\end{minipage}

\vspace{1.2em}

\begin{tikzpicture}[
    node/.style={
        circle,
        draw=black!70,
        line width=0.55pt,
        minimum size=4.5mm,
        inner sep=0pt
      },
    never/.style={node, fill=black!8},
    exposed/.style={node, fill=blue!18, draw=blue!55!black},
    treated/.style={node, fill=red!18, draw=red!60!black},
    early/.style={node, fill=green!20, draw=green!45!black},
    withheld/.style={node, fill=black!8, draw=red!60!black, dashed, line width=0.85pt},
    legend/.style={font=\footnotesize, anchor=west}
  ]
  \node[never] at (0,0) {};
  \node[legend] at (0.35,0) {never-treated, unexposed};
  \node[exposed] at (4.75,0) {};
  \node[legend] at (5.10,0) {never-treated, exposed};
  \node[treated] at (9.20,0) {};
  \node[legend] at (9.55,0) {cohort $g$ treated};

  \node[early] at (2.40,-0.70) {};
  \node[legend] at (2.75,-0.70) {earlier adopter};
  \node[withheld] at (7.15,-0.70) {};
  \node[legend] at (7.50,-0.70) {cohort $g$ withheld};
\end{tikzpicture}
\endgroup
  \endgroup
  \caption{Observed and counterfactual adoption profiles in the example based on a line network. Panel A shows the realized profile at the target date. Panel B withholds cohort $g$ while holding all other adoption times fixed. Units 3 and 6 therefore remain untreated, while unit 6 remains exposed because unit 5 adopted earlier.}
  \label{fig:line_network_rollouts}
\end{figure}

Let $h^1$ and $h^0$ denote generic realized and counterfactual exposure levels.
For any feasible pair $(h^1,h^0)$, the three causal contrasts satisfy
\begin{align*}
  \underbrace{Y_{it}(g,h^1) - Y_{it}(\infty,h^0)}_{\text{total effect}} & = \underbrace{Y_{it}(g,h^1) - Y_{it}(\infty,h^1)}_{\text{switching effect}} + \underbrace{Y_{it}(\infty,h^1) - Y_{it}(\infty,h^0)}_{\text{spillover effect}}.
\end{align*}
The switching term changes own adoption while holding realized exposure fixed.
The spillover term changes exposure from $h^0$ to $h^1$ while the unit remains untreated.
The decomposition holds by definition.

Figure~\ref{fig:effect_decomposition_paths} represents the three paths of potential outcomes in this decomposition at the target date $t=g+l$.
The vertical gaps are causal contrasts.

\begin{figure}[h]
  \centering
  \begingroup
  \def\EffectDecompositionVTwo{}
  \input{tikz/fig_identification_comparisons.tex}
  \endgroup
  \caption{Paths of potential outcomes for the total effect decomposition. Panel A compares $(g,h^1)$ with $(\infty,h^1)$, holding exposure fixed at its realized state. Panel B compares $(\infty,h^1)$ with $(\infty,h^0)$, holding own adoption fixed at infinity. Panel C evaluates both contrasts at the target date $t=g+l$: $(\infty,h^1)$ is the intermediate state between $(g,h^1)$ and $(\infty,h^0)$, so the total vertical gap is the sum of the switching and spillover gaps. Here $h^1$ and $h^0$ are realizations of $H^1_{it}$ and $H^0_{it}$.}
  \label{fig:effect_decomposition_paths}
\end{figure}

Write $\mathbb E_g[\cdot]\equiv\mathbb E[\cdot\mid G_i = g]$ and $\mathbb E_\infty[\cdot]\equiv\mathbb E[\cdot\mid G_i = \infty]$.
Evaluating the preceding contrasts at each target unit's $(H^1_{it},H^0_{it})$ gives the effects for cohort $g$ at event time $l$,
\begin{align*}
  \tau^{TE}(g,l) & \equiv \mathbb E_g\!\left[Y_{it}(g,H^1_{it}) - Y_{it}(\infty,H^0_{it})\right],      \\
  \tau^{SW}(g,l) & \equiv \mathbb E_g\!\left[Y_{it}(g,H^1_{it}) - Y_{it}(\infty,H^1_{it})\right],      \\
  \tau^{SP}(g,l) & \equiv \mathbb E_g\!\left[Y_{it}(\infty,H^1_{it}) - Y_{it}(\infty,H^0_{it})\right].
\end{align*}
The decomposition for each unit gives
\begin{align*}
  \tau^{TE}(g,l) & = \tau^{SW}(g,l) + \tau^{SP}(g,l).
\end{align*}
For the same $(g,l)$ and counterfactual scenario, define the spillover effect among units that never adopt by
\begin{align*}
  \tau_\infty^{SP}(g,l) & \equiv \mathbb E_\infty\!\left[Y_{it}(\infty,H^1_{it}) - Y_{it}(\infty,H^0_{it})\right].
\end{align*}
Own adoption is fixed at infinity for both potential outcomes. Thus, this effect compares realized exposure with the same counterfactual exposure for cohort $g$ used in the decomposition of the total effect.

\begin{example}[Policy effects in the line network example]
  Focus on units 5--7 in Panel A of Figure~\ref{fig:line_network_rollouts}.
  Unit 6 adopts in cohort $g$, unit 5 adopts before $g$, and unit 7 remains untreated.
  Let exposure indicate whether at least one adjacent unit has adopted:
  \begin{align*}
    H^1_{it} = \1\!\left\{\sum_{j:\,j\sim i}\1\{G_j \le t\}>0\right\}.
  \end{align*}
  At $t\ge g$, unit 6 is observed in state $(g,1)$ because unit 5 has already adopted.
  Under the counterfactual profile in Panel B, unit 6 remains exposed in state $(\infty,1)$ because unit 5 retains its earlier adoption, so $h^0=1$.
  If the researcher instead specifies a profile in which no unit ever adopts, its counterfactual state is $(\infty,0)$, so $h^0=0$.
  Only $Y_{6t}(g,1)$ is observed for unit 6; the two missing untreated states $Y_{6t}(\infty,1)$ and $Y_{6t}(\infty,0)$ separate own adoption from exposure.
\end{example}

\begin{remark}
  If exposure does not affect potential outcomes, $Y_{it}(g,h)=Y_{it}(g,0)$ for all $(g,h)$.
  Then $\tau^{SP}(g,l)=0$ and $\tau_\infty^{SP}(g,l)=0$ for any counterfactual exposure level, and
  \begin{align*}
    \tau^{TE}(g,l) & = \tau^{SW}(g,l) = \mathbb E\!\left[Y_{i,g + l}(g,0) - Y_{i,g + l}(\infty,0)\mid G_i = g\right].
  \end{align*}
  The common value is the conventional group--time average treatment effect $ATT(g,g+l)$ in a design with staggered adoption \citep{callaway2021difference}.
  Thus, without interference, the chosen exposure levels do not change the estimand: the total and switching effects both reduce to the conventional effect for cohort $g$ at time $g+l$, and the spillover effect is zero.
\end{remark}

\section{Identification}\label{sec:identification}

Under interference, a conventional DID for cohort $g$ generally differs from the total effect and does not separately identify the switching and spillover components.
Relative to the total effect defined by $H^0_{it}$, it subtracts cohort $g$'s contrast between untreated potential outcomes under realized and counterfactual exposure at the baseline $g-1$ and the change between $g-1$ and $t$ in the corresponding contrast among comparison units that never adopt. It also adds the difference in trends between the target cohort and comparison units under the counterfactual exposure path.
This section establishes identification of the effects defined in Section~\ref{sec:policy-effects} for each cohort and event time under a general counterfactual exposure $H^0_{it}$.
We first impose conditional parallel trends and require that units in cohort $g$ and units that never adopt overlap at the same exposure state.
The comparison at $H^0_{it}$ identifies the total effect and the spillover effect among units that never adopt under the same counterfactual scenario, while the comparison at the realized exposure $H^1_{it}$ identifies the switching effect; the difference between the total and switching effects identifies the spillover effect for cohort $g$.
Parallel trends across exposure states and transportability provide an alternative route to the spillover effect for cohort $g$.
When the conditions for the switching effect also hold, the total effect follows from the decomposition in Section~\ref{sec:policy-effects}.

Throughout this section, $g\in\mathcal G\setminus\{\infty\}$ and $t=g+l\le T$.
We use a common baseline \(t_b\) at which all units are untreated and unexposed under the realized and counterfactual scenarios, so \(H^1_{i,t_b}=H^0_{i,t_b}=0\).
Let $X_i$ denote predetermined covariates.

\subsection{Conventional DID under spillovers}\label{sec:conventional-did}

A conventional DID for cohort $g$ compares the cohort with all units that never adopt, including units exposed in either comparison period.
Panel C of Figure~\ref{fig:effect_decomposition_paths} shows that the path of realized exposure can separate from the path of counterfactual exposure before own adoption; an untreated observation may therefore be exposed.
In the decomposition below, $H^0_{it}$ and $H^0_{i,g-1}$ are induced for every unit at the two dates by the same counterfactual adoption scenario specified by the researcher to define the total effect for cohort $g$.
For $t=g+l$, define
\begin{align*}
  \tau^{DID}(g,l) & \equiv \mathbb E_g[Y_{it} - Y_{i,g-1}] - \mathbb E_\infty[Y_{it} - Y_{i,g-1}].
\end{align*}
This benchmark uses baseline $g-1$ for cohort $g$.
Identification of the total and switching effects uses the common baseline $t_b$.

\begin{prop}[Decomposition of the DID for a cohort]\label{prop:did_decomposition}
  Suppose Assumptions~\ref{assumption:exposure_mapping} and \ref{assumption:consistency} hold and $Y_{i,g-1}(g,H^1_{i,g-1})=Y_{i,g-1}(\infty,H^1_{i,g-1})$ almost surely conditional on $G_i=g$.
  Then
  \begin{align*}
    \tau^{DID}(g,l) & = \underbrace{\mathbb E_g\!\left[Y_{it}(g,H^1_{it}) - Y_{it}(\infty,H^0_{it})\right]}_{\text{total effect}}                                                                                                                                                                          \\
                    & {}- \underbrace{\mathbb E_g\!\left[Y_{i,g-1}(\infty,H^1_{i,g-1}) - Y_{i,g-1}(\infty,H^0_{i,g-1})\right]}_{\text{baseline contrast between untreated outcomes in cohort }g}                                                                                                           \\
                    & {}- \underbrace{\left\{\mathbb E_\infty\!\left[Y_{it}(\infty,H^1_{it}) - Y_{it}(\infty,H^0_{it})\right] - \mathbb E_\infty\!\left[Y_{i,g-1}(\infty,H^1_{i,g-1}) - Y_{i,g-1}(\infty,H^0_{i,g-1})\right]\right\}}_{\text{change in the corresponding contrast among comparison units}} \\
                    & {}+ \underbrace{\left\{\mathbb E_g\!\left[Y_{it}(\infty,H^0_{it}) - Y_{i,g-1}(\infty,H^0_{i,g-1})\right] - \mathbb E_\infty\!\left[Y_{it}(\infty,H^0_{it}) - Y_{i,g-1}(\infty,H^0_{i,g-1})\right]\right\}}_{\text{difference in trends under counterfactual exposure}}.
  \end{align*}
\end{prop}

\begin{proof}
  See Appendix~\ref{app:identification-proofs}.
\end{proof}

\subsection{Identification assumptions}

\begin{assumption}[No anticipation at the common baseline]\label{assumption:no_anticipation}
  For every cohort \(g\in\mathcal G\setminus\{\infty\}\), $Y_{i,t_b}(g,0)=Y_{i,t_b}(\infty,0)$ almost surely conditional on $G_i=g$ and $X_i=x$, for every $x$ in the target support.
\end{assumption}

Assumption~\ref{assumption:no_anticipation} compares the baseline outcome under adoption in cohort $g$ with the outcome under never adopting, holding exposure fixed at zero.
Together with the restriction imposed by the exposure mapping and consistency, it makes the observed baseline outcome for cohort $g$ equal to its untreated and unexposed baseline outcome.
The restriction is confined to the common baseline $t_b$; anticipatory responses may occur at other dates before adoption.

\begin{assumption}[Parallel trends at the counterfactual exposure]\label{assumption:te_parallel_trends}
  For each $(g,l)$, with $t=g+l$, and every $(x,h)$ in the support of $(X_i,H^0_{it})$ conditional on $G_i=g$,
  \begin{align*}
     & \mathbb E\!\left[Y_{it}(\infty,h) - Y_{i,t_b}(\infty,0)\mid G_i = g,\ X_i = x,\ H^0_{it} = h\right]         \\
     & = \mathbb E\!\left[Y_{it}(\infty,h) - Y_{i,t_b}(\infty,0)\mid G_i = \infty,\ X_i = x,\ H^1_{it} = h\right].
  \end{align*}
  For every $(x,h)$ in the support of $(X_i,H^0_{it})$ conditional on $G_i=\infty$,
  \begin{align*}
     & \mathbb E\!\left[Y_{it}(\infty,h) - Y_{i,t_b}(\infty,0)\mid G_i = \infty,\ X_i = x,\ H^0_{it} = h\right]    \\
     & = \mathbb E\!\left[Y_{it}(\infty,h) - Y_{i,t_b}(\infty,0)\mid G_i = \infty,\ X_i = x,\ H^1_{it} = h\right].
  \end{align*}
\end{assumption}

Assumption~\ref{assumption:te_parallel_trends} equates the untreated trend at counterfactual exposure $h$ with the corresponding trend among units that never adopt and are observed at the same covariate value $x$ and exposure state $h$.
The first clause applies to cohort $g$, and the second applies to units that never adopt under the counterfactual scenario for cohort $g$.
Under exposure mapping, consistency, no anticipation where required, and overlap at the relevant exposure state, the observed conditional mean among units that never adopt at $(x,h)$ identifies the missing counterfactual trends in \eqref{eq:te_identification} and \eqref{eq:never_treated_spillover_identification}.
The comparison permits untreated trends to vary across exposure states.

\begin{assumption}[Parallel trends at the realized exposure]\label{assumption:dse_parallel_trends}
  For each $(g,l)$, with $t=g+l$, and every $(x,h)$ in the support of $(X_i,H^1_{it})$ conditional on $G_i=g$,
  \begin{align*}
     & \mathbb E\!\left[Y_{it}(\infty,h) - Y_{i,t_b}(\infty,0)\mid G_i = g,\ X_i = x,\ H^1_{it} = h\right]         \\
     & = \mathbb E\!\left[Y_{it}(\infty,h) - Y_{i,t_b}(\infty,0)\mid G_i = \infty,\ X_i = x,\ H^1_{it} = h\right].
  \end{align*}
\end{assumption}

Assumption~\ref{assumption:dse_parallel_trends} equates cohort $g$'s untreated trend at realized exposure $h$ with the same conditional trend among units that never adopt.
Under exposure mapping, consistency, no anticipation, and overlap at the relevant exposure state, the observed conditional mean at $(x,h)$ identifies the missing untreated trend used for the switching effect.
Untreated trends may vary with $h$.

\begin{assumption}[Overlap]\label{assumption:overlap}
  There exists a constant $\epsilon > 0$, independent of $N$, such that $\Pr (G_i = \infty) > \epsilon$ and $\Pr (G_i = g) > \epsilon$ for every $(g,l)$.
  For every such $(g,l)$, $ \Pr (G_i = \infty,\ H^1_{it} = h\mid X_i = x) > \epsilon$ for almost every $(x,h)$ under the conditional laws $(X_i,H^0_{it})\mid G_i = g$, $(X_i,H^1_{it})\mid G_i = g$, and $(X_i,H^0_{it})\mid G_i = \infty$.
\end{assumption}

Assumption~\ref{assumption:overlap} gives positive mass to cohort $g$ and the group of units that never adopt.
Its conditional support restriction makes the mean outcome change among units that never adopt with $H^1_{it}=h$ available at the $(x,h)$ values over which the total effect, switching effect, and spillover effect among units that never adopt are averaged.

\subsection{Identification results}

The following proposition identifies the total effect from observed outcome changes in cohort $g$ and matched units that never adopt.

\begin{prop}[Total effect identification]\label{prop:dte}
  Fix $(g,l)$ and let $t=g+l$.
  Suppose Assumptions~\ref{assumption:exposure_mapping}, \ref{assumption:consistency}, \ref{assumption:no_anticipation}, \ref{assumption:te_parallel_trends}, and \ref{assumption:overlap} hold.
  Then
  \begin{align}
    \tau^{TE}(g,l) & = \mathbb E_g[Y_{it} - Y_{i,t_b}] - \mathbb E_g\!\left[\sum_{h \in \mathcal H_N}\1\{H^0_{it} = h\}\mathbb E[Y_{it} - Y_{i,t_b}\mid G_i = \infty,\ X_i,\ H^1_{it} = h]\right]. \label{eq:te_identification}
  \end{align}
\end{prop}

\begin{proof}
  See Appendix~\ref{app:identification-proofs}.
\end{proof}

The comparison at $H^0_{it}$ identifies the spillover effect among units that never adopt under the counterfactual scenario for cohort $g$.

\begin{prop}[Identification of the spillover effect among units that never adopt]\label{prop:never_treated_spillover}
  Fix $(g,l)$ and let $t=g+l$.
  Suppose Assumptions~\ref{assumption:exposure_mapping}, \ref{assumption:consistency}, \ref{assumption:te_parallel_trends}, and \ref{assumption:overlap} hold.
  Then
  \begin{align}
    \tau_\infty^{SP}(g,l) & = \mathbb E_\infty[Y_{it}-Y_{i,t_b}] - \mathbb E_\infty\!\left[\sum_{h\in\mathcal H_N}\1\{H^0_{it}=h\}\mathbb E[Y_{it}-Y_{i,t_b}\mid G_i=\infty,\ X_i,\ H^1_{it}=h]\right]. \label{eq:never_treated_spillover_identification}
  \end{align}
\end{prop}

\begin{proof}
  See Appendix~\ref{app:identification-proofs}.
\end{proof}

The second result identifies the switching effect at the realized exposure state.

\begin{prop}[Switching effect identification]\label{prop:dse}
  Fix \((g,l)\) and let \(t=g+l\).
  Suppose Assumptions~\ref{assumption:exposure_mapping}, \ref{assumption:consistency}, \ref{assumption:no_anticipation}, \ref{assumption:dse_parallel_trends}, and \ref{assumption:overlap} hold.
  Then
  \begin{align}
    \tau^{SW}(g,l) & = \mathbb E_g[Y_{it} - Y_{i,t_b}] - \mathbb E_g\!\left[\sum_{h \in \mathcal H_N}\1\{H^1_{it} = h\}\mathbb E[Y_{it} - Y_{i,t_b}\mid G_i = \infty,\ X_i,\ H^1_{it} = h]\right]. \label{eq:sw_identification}
  \end{align}
\end{prop}

\begin{proof}
  See Appendix~\ref{app:identification-proofs}.
\end{proof}

Subtracting \eqref{eq:sw_identification} from \eqref{eq:te_identification} identifies $\tau^{SP}(g,l)$.

\begin{prop}[Spillover effect identification]\label{prop:cse}
  Fix $(g,l)$ and let $t=g+l$.
  Suppose the conditions of Propositions~\ref{prop:dte} and \ref{prop:dse} hold.
  Then
  \begin{align}
    \tau^{SP}(g,l) & = \tau^{TE}(g,l)  - \tau^{SW}(g,l) \notag                                                                                                                                                         \\
                   & = \mathbb E_g\!\left[\sum_{h \in \mathcal H_N}\{\1\{H^1_{it} = h\} - \1\{H^0_{it} = h\}\}\mathbb E[Y_{it} - Y_{i,t_b}\mid G_i = \infty,\ X_i,\ H^1_{it} = h]\right]. \label{eq:sp_identification}
  \end{align}
\end{prop}

\begin{proof}
  For each target unit,
  \begin{align*}
     & \{Y_{it}(g,H^1_{it}) - Y_{it}(\infty,H^0_{it})\} - \{Y_{it}(g,H^1_{it}) - Y_{it}(\infty,H^1_{it})\} \\
     & = Y_{it}(\infty,H^1_{it}) - Y_{it}(\infty,H^0_{it}).
  \end{align*}
  Averaging over cohort $g$ gives $\tau^{SP}(g,l)=\tau^{TE}(g,l)-\tau^{SW}(g,l)$.
  Subtracting the observed-outcome representations in Propositions~\ref{prop:dte} and \ref{prop:dse} cancels the outcome change for cohort $g$ and gives \eqref{eq:sp_identification}.
\end{proof}

\subsection{Alternative identification through exposure contrasts among never-treated units}\label{sec:alternative-identification}

Fix $(g,l)$ and let $t=g+l$.
The preceding results use parallel trends evaluated at the counterfactual and realized exposure states to identify the total and switching effects.
The alternative strategy in this subsection does not impose parallel trends at the counterfactual exposure.
It first identifies the spillover effect for cohort $g$ from exposure contrasts among units that never adopt and a transportability restriction.
Combining this result with the switching effect in Proposition~\ref{prop:dse} then identifies the total effect.

\begin{assumption}[Parallel trends across exposure states]\label{assumption:exposure_parallel_trends}
  For every pair of relevant states $h$ and $h'$ and covariate $x$,
  \begin{align*}
     & \mathbb E[Y_{it}(\infty,0)-Y_{i,t_b}(\infty,0) \mid G_i=\infty,\ X_i=x,\ H^1_{it}=h]     \\
     & = \mathbb E[Y_{it}(\infty,0)-Y_{i,t_b}(\infty,0) \mid G_i=\infty,\ X_i=x,\ H^1_{it}=h'].
  \end{align*}
\end{assumption}

Conditional on $X_i$, Assumption~\ref{assumption:exposure_parallel_trends} equates the trends at zero exposure for units that never adopt and are observed in different exposure states.
With exposure mapping, consistency, and the common unexposed baseline, differences in their observed changes identify differences in untreated exposure effects.
Trends at zero exposure and exposure effects may vary with $X_i$.

\begin{assumption}[Transportability of untreated exposure effects]\label{assumption:exposure_transportability}
  For every $(x,h^1,h^0)$ in the target support,
  \begin{align*}
     & \mathbb E[Y_{it}(\infty,h^1)-Y_{it}(\infty,h^0) \mid G_i=g,\ X_i=x,\ H^1_{it}=h^1,\ H^0_{it}=h^0] \\
     & = \mathbb E[Y_{it}(\infty,h^1)-Y_{it}(\infty,0) \mid G_i=\infty,\ X_i=x,\ H^1_{it}=h^1]           \\
     & - \mathbb E[Y_{it}(\infty,h^0)-Y_{it}(\infty,0) \mid G_i=\infty,\ X_i=x,\ H^1_{it}=h^0].
  \end{align*}
\end{assumption}

Assumption~\ref{assumption:exposure_transportability} imposes conditional mean homogeneity of the untreated exposure contrast from $h^0$ to $h^1$ between cohort $g$ and units that never adopt, given $(x,h^1,h^0)$.
For units that never adopt, the two exposure effects are evaluated in the respective subpopulations observed at $h^1$ and $h^0$.
The assumption does not equate outcome levels or trends at zero exposure across the two populations.
Under exposure mapping, consistency, and the common unexposed baseline, Assumption~\ref{assumption:exposure_parallel_trends} identifies the difference between these effects from observed outcome changes among units that never adopt.
Transportability assigns that difference to cohort $g$.

Assumption~\ref{assumption:overlap} ensures that both target exposure states occur among units that never adopt at covariate value $x$.
For a target pair $(h^1,h^0)$, the formula evaluates conditional means among units that never adopt at $(x,h^1)$ and $(x,h^0)$.

\begin{prop}[Spillover effect identification from untreated exposure contrasts]\label{prop:alternative-identification}
  Suppose Assumptions~\ref{assumption:exposure_mapping}, \ref{assumption:consistency}, \ref{assumption:overlap}, \ref{assumption:exposure_parallel_trends}, and \ref{assumption:exposure_transportability} hold.
  Then
  \begin{align*}
    \tau^{SP}(g,l) & = \mathbb E_g\!\left[\sum_{h \in \mathcal H_N}\{\1\{H^1_{it} = h\} - \1\{H^0_{it} = h\}\}\mathbb E[Y_{it} - Y_{i,t_b}\mid G_i = \infty,\ X_i,\ H^1_{it} = h]\right].
  \end{align*}
  If the conditions of Proposition~\ref{prop:dse} also hold, then
  \begin{align*}
    \tau^{TE}(g,l) & = \left\{\mathbb E_g[Y_{it} - Y_{i,t_b}] - \mathbb E_g\!\left[\sum_{h \in \mathcal H_N}\1\{H^1_{it} = h\}\mathbb E[Y_{it} - Y_{i,t_b}\mid G_i = \infty,\ X_i,\ H^1_{it} = h]\right]\right\} \\
                   & \mathrel{\phantom{=}} + \mathbb E_g\!\left[\sum_{h \in \mathcal H_N}\{\1\{H^1_{it} = h\} - \1\{H^0_{it} = h\}\}\mathbb E[Y_{it} - Y_{i,t_b}\mid G_i = \infty,\ X_i,\ H^1_{it} = h]\right]   \\
                   & = \mathbb E_g[Y_{it} - Y_{i,t_b}] - \mathbb E_g\!\left[\sum_{h \in \mathcal H_N}\1\{H^0_{it} = h\}\mathbb E[Y_{it} - Y_{i,t_b}\mid G_i = \infty,\ X_i,\ H^1_{it} = h]\right].
  \end{align*}
\end{prop}

\begin{proof}
  See Appendix~\ref{app:identification-proofs}.
\end{proof}

Proposition~\ref{prop:cse} uses comparisons under parallel trends between cohort $g$ and units that never adopt at $H^0_{it}$ and $H^1_{it}$.
Proposition~\ref{prop:alternative-identification} uses parallel trends across exposure states among units that never adopt and transportability to cohort $g$.
When the switching conditions also hold, the alternative result gives the same representation of the total effect in terms of observed outcomes as Proposition~\ref{prop:dte}.
Proposition~\ref{prop:never_treated_spillover} matches on exposure and averages the counterfactual trend over units that never adopt, whereas Proposition~\ref{prop:alternative-identification} standardizes exposure contrasts among those units to cohort $g$.

\section{Estimation}\label{sec:estimation}

The identification results in Section~\ref{sec:identification} yield regression estimators by cohort and event time.

\subsection{Outcome regression estimators by cohort and event time}

Fix $(g,l)$ and let $t=g+l$.
Write $\Delta Y_{it}=Y_{it}-Y_{i,t_b}$, $N_g=\sum_{i=1}^N\1\{G_i=g\}$, and $N_\infty=\sum_{i=1}^N\1\{G_i=\infty\}$.
The outcome regression among units that never adopt is
\begin{align*}
  m_t(x,h) & = \mathbb E[\Delta Y_{it}\mid G_i=\infty,\ X_i=x,\ H^1_{it}=h].
\end{align*}
For example, a parametric specification can be estimated by least squares after regressing $\Delta Y_{it}$ among units that never adopt on an intercept, functions of $X_i$, indicators for the exposure states, and selected interactions between the covariates and exposure indicators.
We use a fitted value $\widehat m_t^{TE,g,l}$ for the total effect and the spillover effect among units that never adopt and allow a fit $\widehat m_t^{SW,g,l}$ for the switching effect.
Substituting these fitted regressions into the identifying results gives
\begin{align}
  \widehat\tau_{OR}^{TE}(g,l)        & = \dfrac{1}{N_g}\sum_{i=1}^N\1\{G_i=g\}\{\Delta Y_{it}-\widehat m_t^{TE,g,l}(X_i,H^0_{it})\}, \label{eq:te_or_estimator}                      \\
  \widehat\tau_{OR}^{SW}(g,l)        & = \dfrac{1}{N_g}\sum_{i=1}^N\1\{G_i=g\}\{\Delta Y_{it}-\widehat m_t^{SW,g,l}(X_i,H^1_{it})\}, \label{eq:sw_or_estimator}                      \\
  \widehat\tau_{OR}^{SP}(g,l)        & = \widehat\tau_{OR}^{TE}(g,l)-\widehat\tau_{OR}^{SW}(g,l), \label{eq:sp_or_estimator}                                                         \\
  \widehat\tau_{OR,\infty}^{SP}(g,l) & = \dfrac{1}{N_\infty}\sum_{i=1}^N\1\{G_i=\infty\}\{\Delta Y_{it}-\widehat m_t^{TE,g,l}(X_i,H^0_{it})\}. \label{eq:never_treated_spillover_or}
\end{align}

\subsection{Analysis weights and event-time aggregation}\label{sec:analysis-weights}

Analysis weights define the population over which the causal effects are averaged.
Let $\varpi_i > 0$ be predetermined.
For each cohort used in estimation and for units that never adopt, define the positive masses
\begin{align*}
  s_{\varpi,g}      & = \mathbb E[\varpi_i\1\{G_i=g\}] = \Pr (G_i=g)\mathbb E[\varpi_i\mid G_i=g] > 0,                \\
  s_{\varpi,\infty} & = \mathbb E[\varpi_i\1\{G_i=\infty\}] = \Pr (G_i=\infty)\mathbb E[\varpi_i\mid G_i=\infty] > 0.
\end{align*}
Each mass is the cohort probability multiplied by the mean analysis weight in that cohort.
When the weights are constant, it reduces to that constant multiplied by the cohort probability.
For a fixed $(g,l)$ and $t=g+l$, define
\begin{align*}
  \tau_\varpi^{TE}(g,l)          & \equiv \dfrac{1}{s_{\varpi,g}}\mathbb E\!\left[\varpi_i\1\{G_i = g\}\{Y_{it}(g,H^1_{it}) - Y_{it}(\infty,H^0_{it})\}\right],                \\
  \tau_\varpi^{SW}(g,l)          & \equiv \dfrac{1}{s_{\varpi,g}}\mathbb E\!\left[\varpi_i\1\{G_i = g\}\{Y_{it}(g,H^1_{it}) - Y_{it}(\infty,H^1_{it})\}\right],                \\
  \tau_\varpi^{SP}(g,l)          & \equiv \dfrac{1}{s_{\varpi,g}}\mathbb E\!\left[\varpi_i\1\{G_i = g\}\{Y_{it}(\infty,H^1_{it}) - Y_{it}(\infty,H^0_{it})\}\right],           \\
  \tau_{\varpi,\infty}^{SP}(g,l) & \equiv \dfrac{1}{s_{\varpi,\infty}}\mathbb E\!\left[\varpi_i\1\{G_i = \infty\}\{Y_{it}(\infty,H^1_{it}) - Y_{it}(\infty,H^0_{it})\}\right].
\end{align*}
The decomposition gives $\tau_\varpi^{TE}(g,l)=\tau_\varpi^{SW}(g,l)+\tau_\varpi^{SP}(g,l)$.
The parallel trends and overlap conditions in Section~\ref{sec:identification} are imposed after reweighting the relevant populations by $\varpi_i$.
The weighted versions of \eqref{eq:te_or_estimator}, \eqref{eq:sw_or_estimator}, and \eqref{eq:never_treated_spillover_or} use the corresponding fitted weighted outcome regressions, multiply each summand by $\varpi_i$, and replace $N_g$ and $N_\infty$ by $\sum_{i=1}^N\varpi_i\1\{G_i=g\}$ and $\sum_{i=1}^N\varpi_i\1\{G_i=\infty\}$, respectively.
The weighted spillover estimator for cohort $g$ is the difference between the total and switching estimators.

For aggregation by event time, let $\mathcal L$ be the fixed set of event times and let $\mathcal G_l$ be the set of cohorts included at event time $l$.
The event-time effects weight each cohort in proportion to its total analysis weight.
Define
\begin{align*}
  \omega_g(l) & = \dfrac{s_{\varpi,g}}{\sum_{g'\in\mathcal G_l}s_{\varpi,g'}}
\end{align*}
and its sample analogue as
\begin{align*}
  \widehat\omega_g(l) & = \dfrac{\sum_{i=1}^N\varpi_i\1\{G_i=g\}}{\sum_{g'\in\mathcal G_l}\sum_{i=1}^N\varpi_i\1\{G_i=g'\}}.
\end{align*}
For $Q\in\{TE,SW\}$, the population effects by event time are
\begin{align*}
  \tau_\varpi^Q(l)                 & = \sum_{g\in\mathcal G_l}\omega_g(l)\tau_\varpi^Q(g,l),             \\
  \tau_\varpi^{SP}(l)              & = \tau_\varpi^{TE}(l)-\tau_\varpi^{SW}(l),                          \\
  \bar\tau_{\varpi,\infty}^{SP}(l) & = \sum_{g\in\mathcal G_l}\omega_g(l)\tau_{\varpi,\infty}^{SP}(g,l).
\end{align*}
Replacing the cohort masses and effects by their sample counterparts gives
\begin{align*}
  \widehat\tau_\varpi^Q(l)                   & = \sum_{g\in\mathcal G_l}\widehat\omega_g(l)\widehat\tau_\varpi^Q(g,l),             \\
  \widehat\tau_\varpi^{SP}(l)                & = \widehat\tau_\varpi^{TE}(l)-\widehat\tau_\varpi^{SW}(l),                          \\
  \widehat{\bar\tau}_{\varpi,\infty}^{SP}(l) & = \sum_{g\in\mathcal G_l}\widehat\omega_g(l)\widehat\tau_{\varpi,\infty}^{SP}(g,l).
\end{align*}
Here $\widehat\tau_{\varpi,\infty}^{SP}(g,l)$ is the weighted version of \eqref{eq:never_treated_spillover_or}.

\subsection{Doubly robust estimators}

Fix $(g,l)$, with $t=g+l$.
The doubly robust estimators augment the outcome regression estimators with odds-weighted residuals from units that never adopt.
Define
\begin{align*}
  \lambda_{TE}^{g,l}(x,h) & = \dfrac{\Pr (G_i=g,\ H^0_{it}=h\mid X_i=x)}{\Pr (G_i=\infty,\ H^1_{it}=h\mid X_i=x)}, \\
  \lambda_{SW}^{g,l}(x,h) & = \dfrac{\Pr (G_i=g,\ H^1_{it}=h\mid X_i=x)}{\Pr (G_i=\infty,\ H^1_{it}=h\mid X_i=x)}.
\end{align*}
For $Q\in\{TE,SW\}$, extend $\lambda_Q^{g,l}$ and $\widehat\lambda_Q^{g,l}$ by zero outside the support of the numerator of $\lambda_Q^{g,l}$.
Following \citet{sant2020doubly} and \citet{XU2026106304}, the doubly robust estimators are
\begin{align}
  \widehat\tau_{DR}^{TE}(g,l) & = \dfrac{1}{N_g}\sum_{i=1}^N\1\{G_i=g\}\{\Delta Y_{it}-\widehat m_t^{TE,g,l}(X_i,H^0_{it})\} \notag                                                                   \\
                              & - \dfrac{1}{N_g}\sum_{i=1}^N\1\{G_i=\infty\}\widehat\lambda_{TE}^{g,l}(X_i,H^1_{it})\{\Delta Y_{it}-\widehat m_t^{TE,g,l}(X_i,H^1_{it})\}, \label{eq:te_dr_estimator} \\
  \widehat\tau_{DR}^{SW}(g,l) & = \dfrac{1}{N_g}\sum_{i=1}^N\1\{G_i=g\}\{\Delta Y_{it}-\widehat m_t^{SW,g,l}(X_i,H^1_{it})\}   \notag                                                                 \\
                              & - \dfrac{1}{N_g}\sum_{i=1}^N\1\{G_i=\infty\}\widehat\lambda_{SW}^{g,l}(X_i,H^1_{it})\{\Delta Y_{it}-\widehat m_t^{SW,g,l}(X_i,H^1_{it})\}, \label{eq:sw_dr_estimator} \\
  \widehat\tau_{DR}^{SP}(g,l) & = \widehat\tau_{DR}^{TE}(g,l)-\widehat\tau_{DR}^{SW}(g,l). \label{eq:sp_dr_estimator}
\end{align}
For $Q\in\{TE,SW\}$, let $\widetilde m_t^{Q,g,l}$ and $\widetilde\lambda_Q^{g,l}$ denote the conditional mean and odds functions under the models specified by the researcher, and let $\tau_{DR}^Q(g,l)$ be the corresponding population functional.

\begin{prop}[Double robustness]\label{prop:dr_components}
  Suppose Assumptions~\ref{assumption:exposure_mapping}, \ref{assumption:consistency}, \ref{assumption:no_anticipation}, \ref{assumption:te_parallel_trends}, \ref{assumption:dse_parallel_trends}, and \ref{assumption:overlap} hold.
  For each $Q\in\{TE,SW\}$, if either $\widetilde m_t^{Q,g,l}=m_t$ or $\widetilde\lambda_Q^{g,l}=\lambda_Q^{g,l}$, then $\tau_{DR}^Q(g,l)=\tau^Q(g,l)$.
  If this condition holds for both effects, then $\tau_{DR}^{TE}(g,l)-\tau_{DR}^{SW}(g,l)=\tau^{SP}(g,l)$.
\end{prop}

\begin{proof}
  See Appendix~\ref{app:estimation-details}.
\end{proof}

The weighted versions of \eqref{eq:te_dr_estimator} and \eqref{eq:sw_dr_estimator} multiply each summand by $\varpi_i$, replace $N_g$ by $\sum_{i=1}^N\varpi_i\1\{G_i=g\}$, and use the weighted outcome regressions and odds defined in Appendix~\ref{app:parametric-first-stages}.
Under the corresponding weighted identification conditions, for each $Q\in\{TE,SW\}$, if either $\widetilde m_{\varpi,t}^{Q,g,l}=m_{\varpi,t}$ or $\widetilde\lambda_{\varpi,Q}^{g,l}=\lambda_{\varpi,Q}^{g,l}$, then $\tau_{\varpi,DR}^Q(g,l)=\tau_\varpi^Q(g,l)$.
The weighted doubly robust estimators are aggregated by event time as in Section~\ref{sec:analysis-weights}.

\subsection{Pre-tests}\label{sec:pre-treatment-diagnostics}

Under no anticipation, parallel trends implies testable restrictions before cohort $g$ adopts.
For $t_b<s<g$, define
\begin{align*}
  P^{TE}(g,s) & = \mathbb E_g[Y_{is} - Y_{i,t_b} - m_s(X_i,H^0_{is})], \\
  P^{SW}(g,s) & = \mathbb E_g[Y_{is} - Y_{i,t_b} - m_s(X_i,H^1_{is})].
\end{align*}
The pre-period analogue of parallel trends for the Switching effect implies $P^{SW}(g,s)=0$ under no anticipation at $s$.
The corresponding implication for the Total effect is $P^{TE}(g,s)=0$ when $H^1_{is}=H^0_{is}$ for cohort $g$ at $s$.
When the two exposure states differ, the observed outcome can contain a spillover effect before own adoption, and $P^{TE}(g,s)$ need not be zero under parallel trends for the Total effect.

Under the corresponding support and nuisance conditions at these dates, sample analogues use the same outcome regressions or doubly robust equations and cohort aggregation as the post-adoption estimators.
Weighted diagnostics apply the analysis weights defined in Section~\ref{sec:analysis-weights}.
For a fixed set of pre-treatment dates, we use the joint covariance to test whether the corresponding pre-treatment parameters are zero.
Remark~\ref{app:pretest_inference} states the joint inference result.
Following the pre-treatment pseudo-effects in \citet{callaway2021difference} and the interference diagnostic in \citet{XU2026106304}, these pre-tests assess the plausibility of parallel trends before own adoption.

\section{Inference}\label{sec:inference}

The estimators in Section~\ref{sec:estimation} do not restrict how the nuisance functions are fitted.
For the inference theory, we use finite-dimensional specifications for the conditional means and the propensity scores that induce the odds.
We stack their smooth estimating equations with equations for the effects at each cohort and event time and for the cohort masses.
The joint linearization yields influence functions for the event-time estimators that incorporate estimation error in all of these quantities.
Under spatial dependence, the event-time estimators are jointly asymptotically normal.
Probability limits are under the joint law, whereas weak convergence is conditional on $\mathcal C_N$.
The joint asymptotic normality and spatial HAC covariance results build on \citet{KOJEVNIKOV2021882} and \citet{XU2026106304}.
Appendix~\ref{app:inference-proofs} gives the derivations and proofs, and Appendix~\ref{app:weighted-inference} states the additional conditions for weighted estimators.

\subsection{Spatial sampling and CLT conditions}\label{sec:inference-conditions}

Let $\rho(i,j)$ denote the geographic or network distance used for spatial covariance estimation.
The moment exponent $p>4$ and the regularity conditions on the scores are stated in Appendix~\ref{app:inference-regularity}.

\begin{assumption}[Support and positive masses]\label{assumption:inf_support}
  There exist constants $c_0>0$ and $C_0<\infty$, independent of $N$, such that:
  \begin{enumerate}
    \item[(i)] $|\mathcal L|<\infty$ and $|\mathcal G_l|<\infty$ for every $l\in\mathcal L$.
    \item[(ii)] Assumption~\ref{assumption:overlap} holds for every $l\in\mathcal L$ and $g\in\mathcal G_l$. When DR estimation is used, $0\le \lambda_{TE}^{g,l}(x,h) \le C_0$ and $0\le \lambda_{SW}^{g,l}(x,h) \le C_0$ on the corresponding supports among units that never adopt.
    \item[(iii)] The cohort probabilities satisfy $\inf_{g\in\bigcup_{l\in\mathcal L}\mathcal G_l}\Pr (G_i = g) \ge c_0$ and $\Pr (G_i = \infty) \ge c_0$.
          For every $l\in\mathcal L$, $g\in\mathcal G_l$, and $h\in\mathcal H_N$, if $\Pr (G_i = g,H^0_{i,g + l} = h)+\Pr (G_i = g,H^1_{i,g + l} = h)+\Pr (G_i = \infty,H^0_{i,g + l} = h)>0$, then $\Pr (G_i = \infty,H^1_{i,g + l} = h)\ge c_0$.
  \end{enumerate}
\end{assumption}

Part (i) keeps the number of event-time and cohort coordinates fixed.
Part (ii) keeps the odds used by the DR correction bounded on the relevant supports.
Part (iii) keeps the cohort masses, comparison masses, and aggregation denominators bounded away from zero.

The restrictions in Assumption~\ref{assumption:inf_spatial} on shell growth, increasing domains, and uniform summability are imposed almost surely along the sequence of $\mathcal C_N$-measurable environments.

\begin{assumption}[Spatial sampling and weak dependence]\label{assumption:inf_spatial}
  There exist $C<\infty$ and an integer $d_s\ge1$ such that the following conditions hold.
  \begin{enumerate}
    \item[(i)] The diameter satisfies $\max_{1\le i,j\le N}\rho(i,j) \to\infty$. For every integer $s\ge0$,
          \begin{align*}
            \sup_{i,N}\left|\{j:s \le \rho(i,j) < s + 1\}\right| & \le C(1 + s)^{d_s - 1}.
          \end{align*}
    \item[(ii)] Under the metric $\rho$, the score array underlying the joint influence contributions is conditionally $\psi$-dependent in the sense of Definition~2.2 and Assumption~2.1(a) in \citet*{KOJEVNIKOV2021882}.
          Its dependence coefficients satisfy $0 \le \eta_N(s) \le 1$, where $\eta_N(0)=1$, and
          \begin{align*}
            \sup_N\sum_{s = 1}^{\infty}s^{d_s - 1}\eta_N(s) & < \infty.
          \end{align*}
    \item[(iii)] The scalar spatial CLT rate conditions in Assumption~\ref{assumption:inf_spatial_rates} hold.
  \end{enumerate}
\end{assumption}

Part (i) imposes an increasing-domain sequence while controlling the number of units in each distance shell.
Part (ii) requires dependence to decay sufficiently fast with distance for the accumulated dependence across shells to remain bounded.
Part (iii) imposes the scalar rate conditions used to apply the spatial CLT; these rates are stated in Appendix~\ref{app:spatial-clt-rates}.
Assumption~\ref{assumption:inf_regular} gives consistency and the expansion of the estimating equations.

\subsection{Asymptotic linear representation and joint normality}\label{sec:joint-asymptotic-distribution}

The population influence contributions $\varphi_{i,N}^{TE}(l)$, $\varphi_{i,N}^{SW}(l)$, and $\varphi_{i,N}^{\infty,SP}(l)$ are defined in Appendix~\ref{app:inference-regularity}.
The following theorem establishes the asymptotic linear representation before imposing the spatial CLT.

\begin{theorem}[Asymptotic linear representation]\label{thm:event_time_distribution}
  Suppose Assumptions~\ref{assumption:inf_support} and \ref{assumption:inf_regular} hold, and the population counterparts of the total estimator, the switching estimator, and the spillover estimator among units that never adopt equal their corresponding event-time targets.
  Then, for $Q\in\{\mathrm{TE},\mathrm{SW}\}$ and $l\in\mathcal L$, the corresponding coordinate satisfies
  \begin{align}
    \sqrt N\{\widehat\tau^Q(l) - \tau^Q(l)\} & = \dfrac{1}{\sqrt N}\sum_{i=1}^N\varphi_{i,N}^Q(l) + o_p(1). \label{eq:joint_influence_representation}
  \end{align}
  The coordinate for the spillover effect among units that never adopt satisfies
  \begin{align*}
    \sqrt N\{\widehat{\bar\tau}_{\infty}^{SP}(l) - \bar\tau_{\infty}^{SP}(l)\} & = \dfrac{1}{\sqrt N}\sum_{i=1}^N\varphi_{i,N}^{\infty,SP}(l) + o_p(1).
  \end{align*}
\end{theorem}

Let $\boldsymbol\zeta$ stack, in order, the event-time Total, Switching, and Spillover effects for the target cohorts and the event-time Spillover effects among units that never adopt.
Let $\widehat{\boldsymbol\zeta}$ stack their estimators in the same order.
Let $\Gamma$ denote the covariance matrix of the Gaussian vector in the following corollary.
For $Q,Q'\in\{\mathrm{TE},\mathrm{SW},\mathrm{SP}\}$, write $\Gamma_{QQ'}(l,l')$ for its entry corresponding to effects $Q$ and $Q'$ at event times $l$ and $l'$.
Write $\Gamma_{\infty}^{SP}(l,l')$ for the entry corresponding to two spillover effects among units that never adopt.

\begin{cor}[Joint asymptotic normality]\label{cor:joint_asymptotic_normality}
  Suppose the conditions of Theorem~\ref{thm:event_time_distribution} and Assumption~\ref{assumption:inf_spatial} hold.
  Then
  \begin{align}
    \sqrt N(\widehat{\boldsymbol\zeta} - \boldsymbol\zeta) & \xrightarrow{d}\mathcal N(0,\Gamma). \label{eq:joint_effect_limit}
  \end{align}
\end{cor}

The influence contribution for the spillover effect for the target cohorts and its covariance follow from the difference between the total and switching effects:
\begin{align*}
  \varphi_{i,N}^{SP}(l) & = \varphi_{i,N}^{TE}(l)-\varphi_{i,N}^{SW}(l),                                               \\
  \Gamma_{SP,SP}(l,l')  & = \Gamma_{TE,TE}(l,l') + \Gamma_{SW,SW}(l,l') - \Gamma_{TE,SW}(l,l') - \Gamma_{SW,TE}(l,l').
\end{align*}
The asymptotic arguments also apply to a fixed cohort and event-time pair by selecting the corresponding cohort coordinate instead of aggregating over cohorts.
This cohort-specific result uses the support, centering, spatial CLT, and spatial HAC conditions for that coordinate, as stated in Appendix~\ref{app:cohort_specific_inference}.
Inference for cohort and group-time effects before event-time aggregation is also considered by \citet{callaway2021difference} and \citet{sun2021estimating}.

\subsection{Spatial HAC covariance and confidence intervals}\label{sec:spatial-hac}

Appendix~\ref{app:inference-regularity} defines the feasible influence contribution $\widehat\varphi_{i,N}^Q(l)$ for each jointly estimated effect.
For any pair formed from the total effect, the switching effect, and the spillover effect among units that never adopt, define
\begin{align}
  \widehat\Gamma_{QQ'}(l,l') & = \dfrac{1}{N}\sum_{i=1}^N\sum_{j=1}^N K\!\left(\dfrac{\rho(i,j)}{b_N}\right)\widehat\varphi_{i,N}^Q(l)\widehat\varphi_{j,N}^{Q'}(l'). \label{eq:shac_event_cov}
\end{align}
For the spillover effect for the target cohorts under no own adoption, the same formula uses $\widehat\varphi_{i,N}^{SP}(l)=\widehat\varphi_{i,N}^{TE}(l)-\widehat\varphi_{i,N}^{SW}(l)$.
Write $\widehat\Gamma_{\infty}^{SP}(l,l')$ for the covariance estimator obtained from \eqref{eq:shac_event_cov} using two influence coordinates for the spillover effect among units that never adopt.
Stacking these covariance estimators in the order defining $\widehat{\boldsymbol\zeta}$ gives $\widehat\Gamma$.
Here $K$ is the spatial kernel and $b_N$ is the bandwidth.
This estimator follows the spatial HAC construction of \citet{conley1999gmm}; \citet{KOJEVNIKOV2021882} gives a corresponding HAC result for network dependence.
For geographic locations embedded in $\mathbb R^3$, one choice is the Wendland $C^2$ kernel, $K(u)=(1-u)^4(1+4u)$ for $0\le u<1$ and $K(u)=0$ for $u\ge1$.
With chordal distance, this radial kernel produces a positive-semidefinite weight matrix \citep{wendland1995piecewise}.
The same conclusion does not generally follow for graph or great-circle distance.
The Bartlett kernel used in the Monte Carlo simulation is paired with a block distance construction for which the kernel matrix is positive semidefinite; a Bartlett kernel need not produce positive-semidefinite weights for an arbitrary graph distance.

\begin{assumption}[Spatial HAC covariance estimation]\label{assumption:inf_shac}
  The kernel, bandwidth, and score dependence satisfy the following conditions:
  \begin{enumerate}
    \item[(i)] $K(0)=1$, $\sup_u|K(u)|<\infty$, and $K(u)=0$ for $u>1$.
    \item[(ii)] The bandwidth satisfies $ b_N \to\infty$ and $b_N  =o\!\left(N^{1/(2d_s)}\right)$.
    \item[(iii)] The kernel truncation error satisfies
          \begin{align*}
            \sum_{s = 0}^{\infty} \sup_{u\in[s,s + 1)}|K(u/b_N) - 1| (1 + s)^{d_s - 1}\eta_N(s)^{1 - 2/p} & \xrightarrow{p} 0.
          \end{align*}
    \item[(iv)] The kernel weight matrix is positive semidefinite with probability approaching one.
    \item[(v)] The covariance rate and the conditions for influence replacement and centering in Assumption~\ref{assumption:inf_shac_technical} hold.
  \end{enumerate}
\end{assumption}

Part (i) normalizes, bounds, and truncates the kernel.
Part (ii) allows the bandwidth to increase while keeping each spatial neighborhood small relative to the sample.
Part (iii) makes truncation asymptotically negligible under the dependence decay in Assumption~\ref{assumption:inf_spatial}.
Part (iv) is the condition used for conservative variance estimation.
Part (v) controls estimation of the influence contributions and the remaining covariance terms.
The covariance rate used in the proof is stated in Appendix~\ref{app:shac-technical}.
The kernel and bandwidth conditions in that assumption are distinct from the positive-semidefiniteness condition in part (iv).

\begin{prop}[Spatial HAC covariance]\label{prop:shac_covariance}
  Suppose Assumptions~\ref{assumption:inf_support}, \ref{assumption:inf_regular}, \ref{assumption:inf_spatial}, and \ref{assumption:inf_shac} hold.
  For every fixed $l\in\mathcal L$ and $Q\in\{\mathrm{TE},\mathrm{SW},\mathrm{SP}\}$, if $\Gamma_{QQ}(l,l)>0$, then, for every $\epsilon>0$,
  \begin{align*}
    \Pr (\widehat\Gamma_{QQ}(l,l)\ge\Gamma_{QQ}(l,l)-\epsilon\mid\mathcal C_N) & \xrightarrow{p} 1.
  \end{align*}
  For every fixed $l\in\mathcal L$, if $\Gamma_{\infty}^{SP}(l,l)>0$, then, for every $\epsilon>0$,
  \begin{align*}
    \Pr (\widehat\Gamma_{\infty}^{SP}(l,l)\ge\Gamma_{\infty}^{SP}(l,l)-\epsilon\mid\mathcal C_N) & \xrightarrow{p} 1.
  \end{align*}
  If $\mathbb E[\varphi_{i,N}^Q(l)\mid\mathcal C_N]=0$ for every $i=1,\ldots,N$, $l\in\mathcal L$, and $Q\in\{\mathrm{TE},\mathrm{SW}\}$, and $\mathbb E[\varphi_{i,N}^{\infty,SP}(l)\mid\mathcal C_N]=0$ for every $i=1,\ldots,N$ and $l\in\mathcal L$, then $\|\widehat\Gamma-\Gamma\|=o_p(1)$.
\end{prop}

The covariance decomposition is given in Appendix~\ref{app:shac-technical}, and the proof of Proposition~\ref{prop:shac_covariance} is given in Appendix~\ref{app:inference-proofs}.
We construct pointwise confidence intervals for each effect.
For $Q\in\{\mathrm{TE},\mathrm{SW},\mathrm{SP}\}$, the pointwise interval is
\begin{align*}
  \widehat C_Q^{pt}(l) & = \left[\widehat\tau^Q(l) \pm z_{1 - \alpha/2}\sqrt{\dfrac{\widehat\Gamma_{QQ}(l,l)}{N}}\right].
\end{align*}
The interval for the spillover effect among units that never adopt is
\begin{align*}
  \widehat C_{\infty,SP}^{pt}(l) & = \left[\widehat{\bar\tau}_{\infty}^{SP}(l) \pm z_{1 - \alpha/2}\sqrt{\dfrac{\widehat\Gamma_{\infty}^{SP}(l,l)}{N}}\right].
\end{align*}
\begin{cor}[Pointwise confidence intervals]\label{cor:pointwise_ci}
  Fix $Q\in\{\mathrm{TE},\mathrm{SW},\mathrm{SP}\}$ and $l\in\mathcal L$.
  Under the conditions of Corollary~\ref{cor:joint_asymptotic_normality} and Proposition~\ref{prop:shac_covariance}, suppose $\Gamma_{QQ}(l,l)>0$.
  Then $\Pr (\tau^Q(l)\in\widehat C_Q^{pt}(l)\mid\mathcal C_N)\ge 1-\alpha+o_p(1)$.
  If $\mathbb E[\varphi_{i,N}^Q(l)\mid\mathcal C_N]=0$ for every $i$, then $\Pr (\tau^Q(l)\in\widehat C_Q^{pt}(l)\mid\mathcal C_N)=1-\alpha+o_p(1)$.
  If $\Gamma_{\infty}^{SP}(l,l)>0$, then $\Pr (\bar\tau_{\infty}^{SP}(l)\in\widehat C_{\infty,SP}^{pt}(l)\mid\mathcal C_N)\ge 1-\alpha+o_p(1)$.
  If, in addition, $\mathbb E[\varphi_{i,N}^{\infty,SP}(l)\mid\mathcal C_N]=0$ for every $i$, then $\Pr (\bar\tau_{\infty}^{SP}(l)\in\widehat C_{\infty,SP}^{pt}(l)\mid\mathcal C_N)=1-\alpha+o_p(1)$.
\end{cor}

\section{Monte Carlo simulation}\label{sec:simulation}

This section evaluates the finite-sample performance of the outcome regression and doubly robust estimators for the total, switching, and spillover effects by event time.
Each estimator is evaluated relative to its target for the finite population under the counterfactual in which no unit adopts.
DGP 1 and DGP 2 evaluate the outcome regression estimators under correctly specified outcome regressions and distinct network and exposure structures.
DGP 3 misspecifies the switching outcome regression while correctly specifying the relevant odds; comparing the outcome regression and doubly robust estimators therefore isolates the contribution of the residual correction.
We report bias, RMSE, and coverage of pointwise 95 percent confidence intervals separately at event times $l=0,1,2$ for sample sizes $N\in\{200,500\}$, using 1,000 replications and equal analysis weights.

\subsection{Design}

The three designs use periods $t=1,\ldots,6$, treated cohorts $3$, $4$, and $5$, and a group that never adopts.
All units are untreated and unexposed in period 1.
Both primitive estimators therefore use the common baseline $t_b=1$.
The event-time averages use cohorts $\{3,4,5\}$ at $l=0,1$ and cohorts $\{3,4\}$ at $l=2$.
We stop at $l=2$ because only cohort 3 would contribute to the event-time target at $l=3$.

The raw exposure index is the share of adopted neighbors, with each row of weights normalized to one: $\widetilde H_{it} = \sum_{j\neq i}w_{ij}\1\{G_j\le t\}$.
DGP 1 and DGP 3 classify exposure as zero, low, or high according as $\widetilde H_{it}=0$, $0<\widetilde H_{it}\le0.5$, or $\widetilde H_{it}>0.5$.
DGP 2 uses $0.55$ in place of $0.5$.
The corresponding dose scores are $q(0)=0$, $q(\mathrm{low})=1$, and $q(\mathrm{high})=2$.

Let $B_i=\1\{X_i>0\}$.
The potential outcomes are
\begin{align*}
  Y_{it}(\infty,h) & = \phi_i + 0.2(t - 1) + 0.5X_i + \{\xi_t + \xi_BB_i\}q(h) + \varepsilon_{it}, \\
  Y_{it}(g,h)      & = Y_{it}(\infty,h) + \tau_l + 0.4q(h), \quad t = g + l \ge g.
\end{align*}
The unit effects $\phi_i$ and disturbances $\varepsilon_{it}$ independently follow $N(0,1)$.
For DGP 1 and DGP 2, $X_i$ is also independently drawn from $N(0,1)$.
We set $\xi_t=0$ for $t\le2$ and $\xi_t=0.3$ for $t\ge3$, and set $\tau_0=1$, $\tau_1=1.5$, and $\tau_l=2$ for $l\ge2$.
The coefficient on the interaction between $B_i$ and exposure is $\xi_B=0$ in DGP 1 and DGP 2 and $\xi_B=0.8$ in DGP 3.

DGP 1 and DGP 3 use blocks of 20 units containing five units from each adoption group.
Every unit has four directed neighbors within its block, each with weight $1/4$.
For exposure classes A, B, E, and H, the neighbor counts from cohorts $(3,4,5,\infty)$ are, respectively, $(0,0,0,4)$, $(1,0,0,3)$, $(1,1,0,2)$, and $(1,1,1,1)$.
DGP 1 assigns one quarter of each exposure class to each adoption group.
The class counts $(\mathrm{A},\mathrm{B},\mathrm{E},\mathrm{H})$ are $(60,60,40,40)$ at $N=200$ and $(140,120,120,120)$ at $N=500$.
For every cohort and event time entering the averages, these assignments provide units that never adopt both in the pure control state and at the exposure states realized in the target cohort.

DGP 2 partitions units into complete regions of size 20, with no links across regions.
Four fifths of the regions contain policy adoption, with five units from each treated cohort and five units that never adopt.
The remaining regions contain 20 units that never adopt.
Thus, the cohort shares are $(0.20,0.20,0.20,0.40)$ for $(3,4,5,\infty)$, and the inactive regions supply pure controls.

DGP 3 retains the DGP 1 neighbor compositions but sets the class shares $(\mathrm{A},\mathrm{B},\mathrm{E},\mathrm{H})$ to $(0.40,0.20,0.20,0.20)$ at both sample sizes.
Conditional on cohort and independently of exposure class, $B_i$ follows a Bernoulli distribution with probability $0.60$ for treated cohorts and $0.40$ for units that never adopt, and $X_i=2B_i-1$.

In DGP 1 and DGP 2, the total effect uses the sample mean of $Y_{it}-Y_{i1}$ in the pure control group, and the switching regression is saturated in $H^1_{it}$ among units that never adopt.
Both outcome regressions are correctly specified under the outcome model.
In DGP 3, the total regression includes an intercept and $B_i$.
The switching regression is saturated in $H^1_{it}$ but omits its interaction with $B_i$.
The DR estimator uses separate total and switching odds of the form $\exp(\gamma_0 + \gamma_1B_i)$.
These odds are correctly specified, so DGP 3 isolates misspecification of the switching outcome regression.
The total regression and DR estimates coincide by construction.

For $Q\in\{\mathrm{TE},\mathrm{SW},\mathrm{SP}\}$, let $\tau_r^Q(l)$ denote the realized target for the finite population in replication $r$ and let $\widehat C_{r,Q}^{pt}(l)$ denote its pointwise 95 percent interval.
With $R = 1000$ replications, the performance measures are
\begin{align*}
  \operatorname{Bias}_Q(l)     & = \dfrac{1}{R}\sum_{r = 1}^R \{\widehat\tau_r^Q(l) - \tau_r^Q(l)\},                        \\
  \operatorname{RMSE}_Q(l)     & = \left[ \dfrac{1}{R}\sum_{r = 1}^R \{\widehat\tau_r^Q(l) - \tau_r^Q(l)\}^2 \right]^{1/2}, \\
  \operatorname{Coverage}_Q(l) & = \dfrac{1}{R}\sum_{r = 1}^R \1\{\tau_r^Q(l) \in \widehat C_{r,Q}^{pt}(l)\}.
\end{align*}
The spillover estimate is the total estimate minus the switching estimate in the same cells defined by cohort and event time.
Pointwise intervals use the joint GMM influence contributions with a Bartlett spatial HAC kernel, block distance, and a bandwidth of three.

\subsection{Results}

Table~\ref{tab:sim_proposed_componentwise_v2} reports the outcome regression estimators in DGP 1 and DGP 2.
Across the three event times, absolute bias does not exceed 0.018, and every RMSE is smaller at $N=500$ than at $N=200$.
Pointwise coverage ranges from 0.883 to 0.943; the largest shortfall occurs for the DGP 2 total effect at $l=0$ and $N=200$.

\begin{table}[h!]
  \centering
  \caption{Finite-sample performance of the outcome regression estimators by event time}
  \label{tab:sim_proposed_componentwise_v2}
  \begin{threeparttable}
\small
\begin{tabular}{llcccccc}
\toprule
Effect & $l$ & \multicolumn{3}{c}{$N=200$} & \multicolumn{3}{c}{$N=500$} \\
\cmidrule(lr){3-5}\cmidrule(lr){6-8}
 & & Bias & RMSE & Coverage & Bias & RMSE & Coverage \\
\midrule
\multicolumn{8}{l}{Panel A. DGP 1} \\
Total & 0 & 0.015 & 0.321 & 0.904 & 0.003 & 0.211 & 0.934 \\
 & 1 & 0.013 & 0.333 & 0.903 & 0.001 & 0.211 & 0.933 \\
 & 2 & 0.014 & 0.365 & 0.910 & -0.002 & 0.231 & 0.933 \\
\addlinespace
Switching & 0 & 0.008 & 0.193 & 0.928 & -0.001 & 0.127 & 0.934 \\
 & 1 & 0.005 & 0.201 & 0.922 & -0.002 & 0.125 & 0.937 \\
 & 2 & 0.008 & 0.228 & 0.927 & -0.006 & 0.144 & 0.943 \\
\addlinespace
Spillover & 0 & 0.006 & 0.260 & 0.902 & 0.005 & 0.170 & 0.921 \\
 & 1 & 0.008 & 0.255 & 0.911 & 0.003 & 0.167 & 0.930 \\
 & 2 & 0.007 & 0.264 & 0.915 & 0.004 & 0.175 & 0.942 \\
\midrule
\multicolumn{8}{l}{Panel B. DGP 2} \\
Total & 0 & 0.007 & 0.223 & 0.883 & 0.003 & 0.145 & 0.921 \\
 & 1 & -0.002 & 0.222 & 0.886 & 0.004 & 0.143 & 0.919 \\
 & 2 & 0.003 & 0.248 & 0.895 & 0.004 & 0.160 & 0.922 \\
\addlinespace
Switching & 0 & 0.004 & 0.220 & 0.924 & 0.006 & 0.146 & 0.928 \\
 & 1 & 0.005 & 0.214 & 0.931 & 0.006 & 0.146 & 0.929 \\
 & 2 & 0.018 & 0.246 & 0.910 & 0.004 & 0.161 & 0.941 \\
\addlinespace
Spillover & 0 & 0.003 & 0.246 & 0.912 & -0.004 & 0.165 & 0.935 \\
 & 1 & -0.007 & 0.245 & 0.915 & -0.002 & 0.164 & 0.930 \\
 & 2 & -0.015 & 0.263 & 0.911 & 0.000 & 0.176 & 0.924 \\
\bottomrule
\end{tabular}
\begin{tablenotes}[flushleft]
\footnotesize
\item Notes: The table reports the outcome regression estimators. Bias, RMSE, and pointwise coverage are computed separately at each event time relative to the realized target for the finite population. Confidence intervals use spatial HAC covariance estimates. Each design, sample size, and event time uses 1,000 Monte Carlo replications.
\end{tablenotes}
\end{threeparttable}

\end{table}

Table~\ref{tab:sim_d3_dr_robustness_v2} compares the outcome regression and DR estimators in DGP 3.
For the switching effect, the bias of the outcome regression estimator ranges from 0.109 to 0.138 at $N=200$ and from 0.107 to 0.130 at $N=500$ across $l=0,1,2$.
The largest absolute DR bias over these event times is 0.015 and 0.004, respectively.
For the spillover effect, the corresponding largest absolute DR biases are 0.007 and 0.006.
The DR correction also moves pointwise coverage closer to 0.95 at each $l\in\{0,1,2\}$ for both effects.
The total effect regression is correctly specified, so the two total effect estimates coincide.

\begin{table}[h!]
  \centering
  \caption{Nuisance misspecification and DR correction in DGP 3}
  \label{tab:sim_d3_dr_robustness_v2}
  \begin{threeparttable}
\small
\begin{tabular}{lllcccccc}
\toprule
Effect & $l$ & Estimator & \multicolumn{3}{c}{$N=200$} & \multicolumn{3}{c}{$N=500$} \\
\cmidrule(lr){4-6}\cmidrule(lr){7-9}
 & & & Bias & RMSE & Coverage & Bias & RMSE & Coverage \\
\midrule
Total & 0 & Outcome regression & 0.008 & 0.317 & 0.902 & 0.006 & 0.192 & 0.930 \\
 &  & DR & 0.008 & 0.317 & 0.902 & 0.006 & 0.192 & 0.930 \\
 & 1 & Outcome regression & 0.012 & 0.325 & 0.895 & 0.006 & 0.194 & 0.935 \\
 &  & DR & 0.012 & 0.325 & 0.895 & 0.006 & 0.194 & 0.935 \\
 & 2 & Outcome regression & 0.008 & 0.349 & 0.908 & 0.006 & 0.206 & 0.942 \\
 &  & DR & 0.008 & 0.349 & 0.908 & 0.006 & 0.206 & 0.942 \\
\addlinespace
Switching & 0 & Outcome regression & 0.109 & 0.238 & 0.893 & 0.107 & 0.167 & 0.867 \\
 &  & DR & 0.006 & 0.223 & 0.912 & 0.000 & 0.129 & 0.952 \\
 & 1 & Outcome regression & 0.118 & 0.247 & 0.878 & 0.120 & 0.176 & 0.852 \\
 &  & DR & 0.008 & 0.227 & 0.916 & 0.002 & 0.131 & 0.951 \\
 & 2 & Outcome regression & 0.138 & 0.276 & 0.872 & 0.130 & 0.192 & 0.863 \\
 &  & DR & 0.015 & 0.249 & 0.920 & 0.004 & 0.144 & 0.955 \\
\addlinespace
Spillover & 0 & Outcome regression & -0.100 & 0.263 & 0.878 & -0.101 & 0.185 & 0.861 \\
 &  & DR & 0.002 & 0.226 & 0.921 & 0.006 & 0.146 & 0.935 \\
 & 1 & Outcome regression & -0.106 & 0.275 & 0.873 & -0.114 & 0.195 & 0.848 \\
 &  & DR & 0.004 & 0.238 & 0.908 & 0.004 & 0.149 & 0.933 \\
 & 2 & Outcome regression & -0.130 & 0.304 & 0.862 & -0.124 & 0.209 & 0.848 \\
 &  & DR & -0.007 & 0.257 & 0.906 & 0.002 & 0.158 & 0.939 \\
\bottomrule
\end{tabular}
\begin{tablenotes}[flushleft]
\footnotesize
\item Notes: DGP 3 misspecifies the switching outcome regression while keeping the corresponding odds correctly specified. Bias, RMSE, and pointwise coverage are computed separately at each event time relative to the realized target for the finite population. Each sample size and event time uses 1,000 Monte Carlo replications. The outcome regression for the total effect is correctly specified, and the outcome regression and DR total estimates coincide by construction.
\end{tablenotes}
\end{threeparttable}

\end{table}

\section{Application}\label{sec:application}

We study the rollout of Community Health Centers.
The outcome is the mortality rate from all causes among adults aged 50 and older, adjusted for age.
We estimate the total, switching, and spillover effects and compare the total effect with Standard DID estimates that do not adjust for exposure.

\subsection{Data and exposure}

Community Health Centers may affect nearby counties through travel, provider networks, and referrals.
Early Community Health Center programs operated through multiple clinics, satellite stations, mobile units, home care, and transportation \citep{bailey2015war}.
We estimate the total, switching, and spillover effects for treated counties and the spillover effect among never-treated counties.

We use the panel of counties observed from 1959 through 1978 in the replication files for \citet{bailey2015war}.
The estimation sample contains 114 counties first treated in 1965--1974 and 2,423 never-treated counties.
The outcome is the mortality rate from all causes among adults aged 50 and older, adjusted for age and expressed per 100,000 residents.
The 1960 population in this age group defines the analysis weights.
Following the replication sample, we exclude Los Angeles, Cook, and New York counties.

Treatment is absorbing, event time is $l=t-g$, and both primitive comparisons use 1959 as the common baseline.
The exposure universe is the full county panel after these exclusions; counties first treated after 1974 contribute to exposure but not to the target cohorts or the never-treated comparison group.
Let $d_{ij}$ denote the distance between representative interior points of counties, constructed from the 2015 Census TIGER/Line shapefile.
A county is exposed when another county within 75 miles has adopted:
\begin{align}
  H^1_{it} & = \1\left\{\sum_{j \ne i}\1\{d_{ij} \le 75\}\1\{t \ge G_j\} > 0\right\}. \label{eq:application_exposure}
\end{align}

Figure~\ref{fig:application_spatial_spillover_exposure} displays realized states of own adoption and exposure under the mapping with a 75-mile radius.

\begin{figure}[h]
  \centering
  \begin{minipage}{0.48\textwidth}
    \centering
    \includegraphics[width=\linewidth]{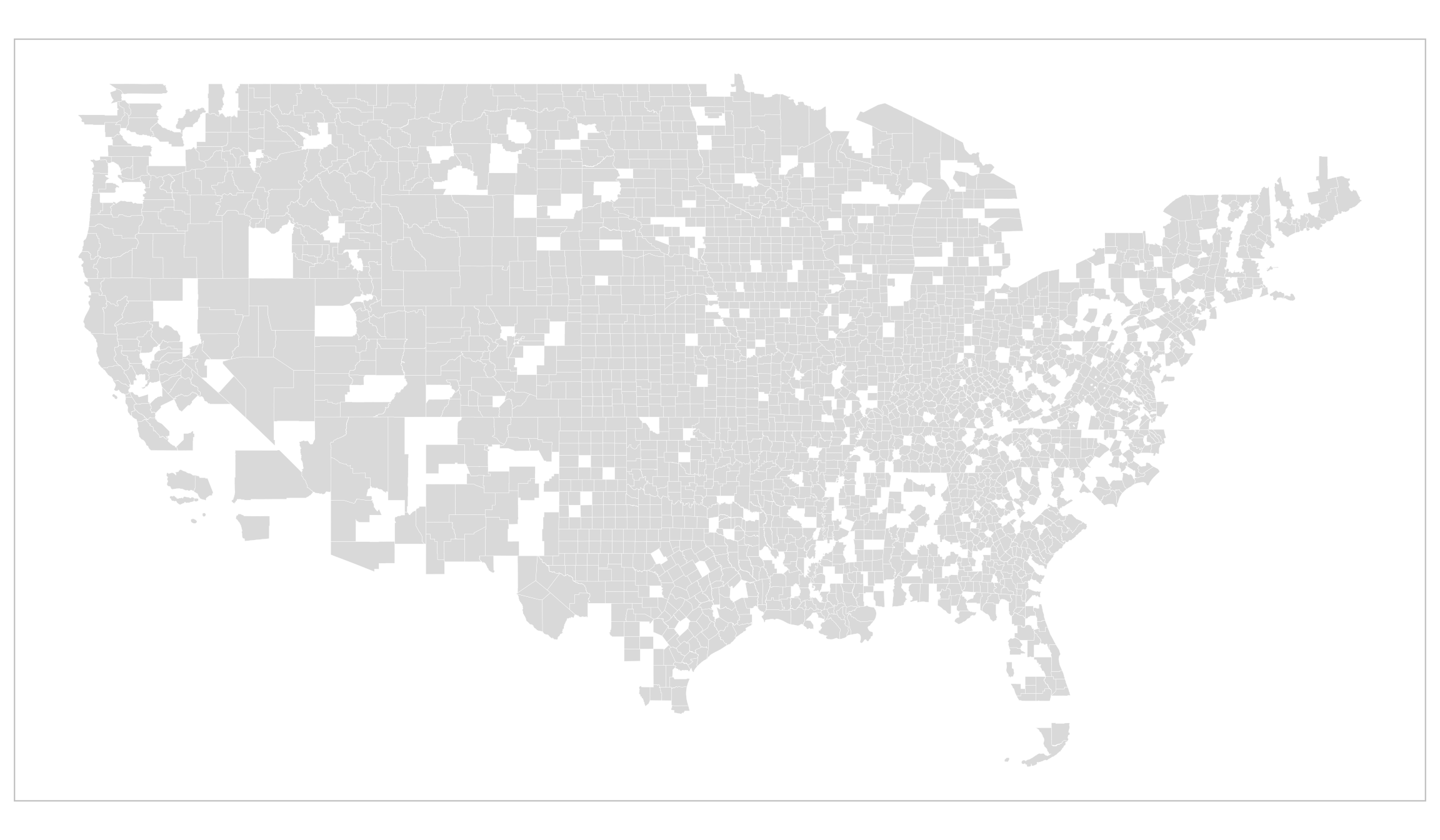}
    {\footnotesize (a) 1964}
  \end{minipage}
  \hfill
  \begin{minipage}{0.48\textwidth}
    \centering
    \includegraphics[width=\linewidth]{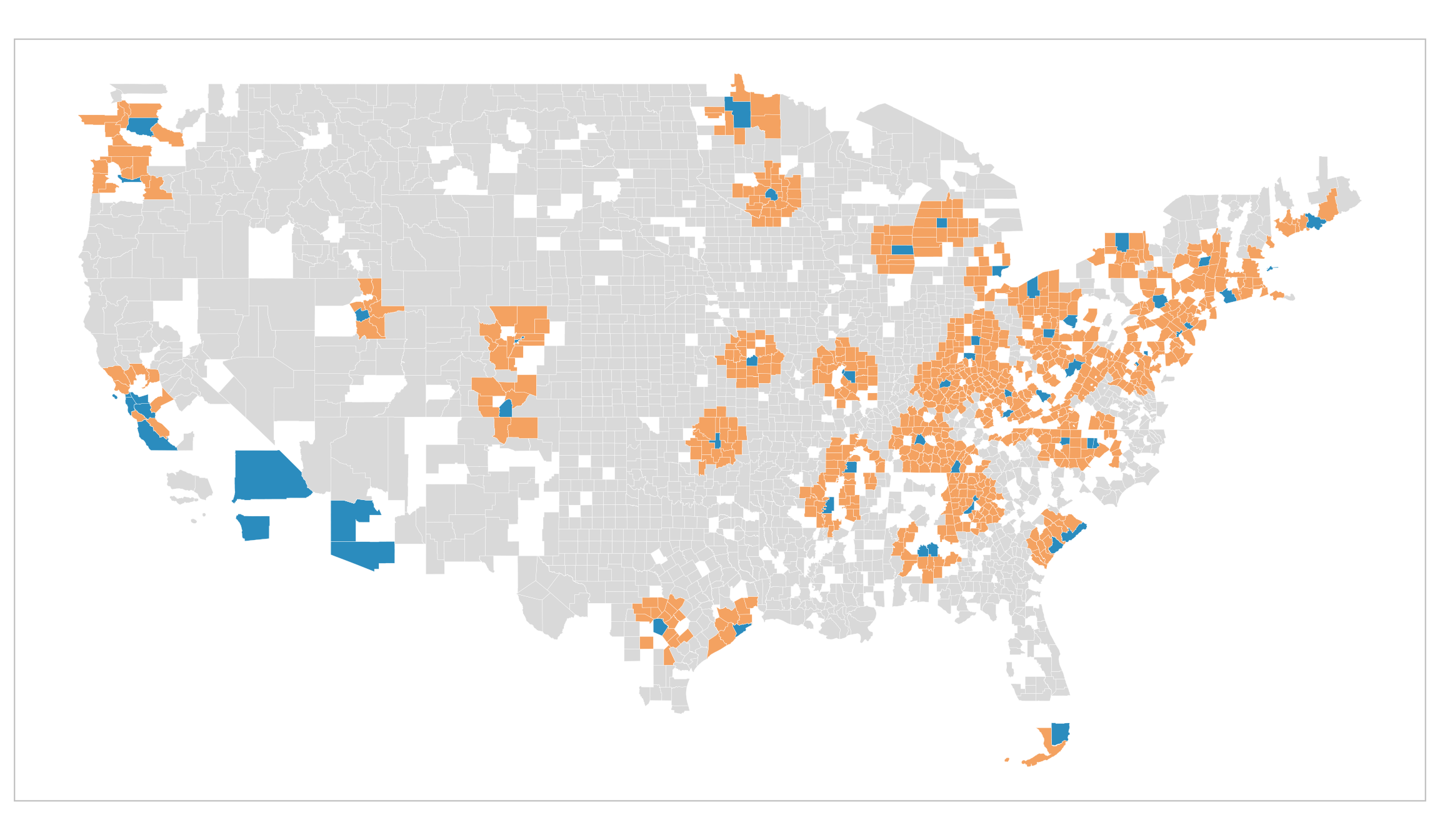}
    {\footnotesize (b) 1970}
  \end{minipage}
  \par\medskip
  \begin{minipage}{0.48\textwidth}
    \centering
    \includegraphics[width=\linewidth]{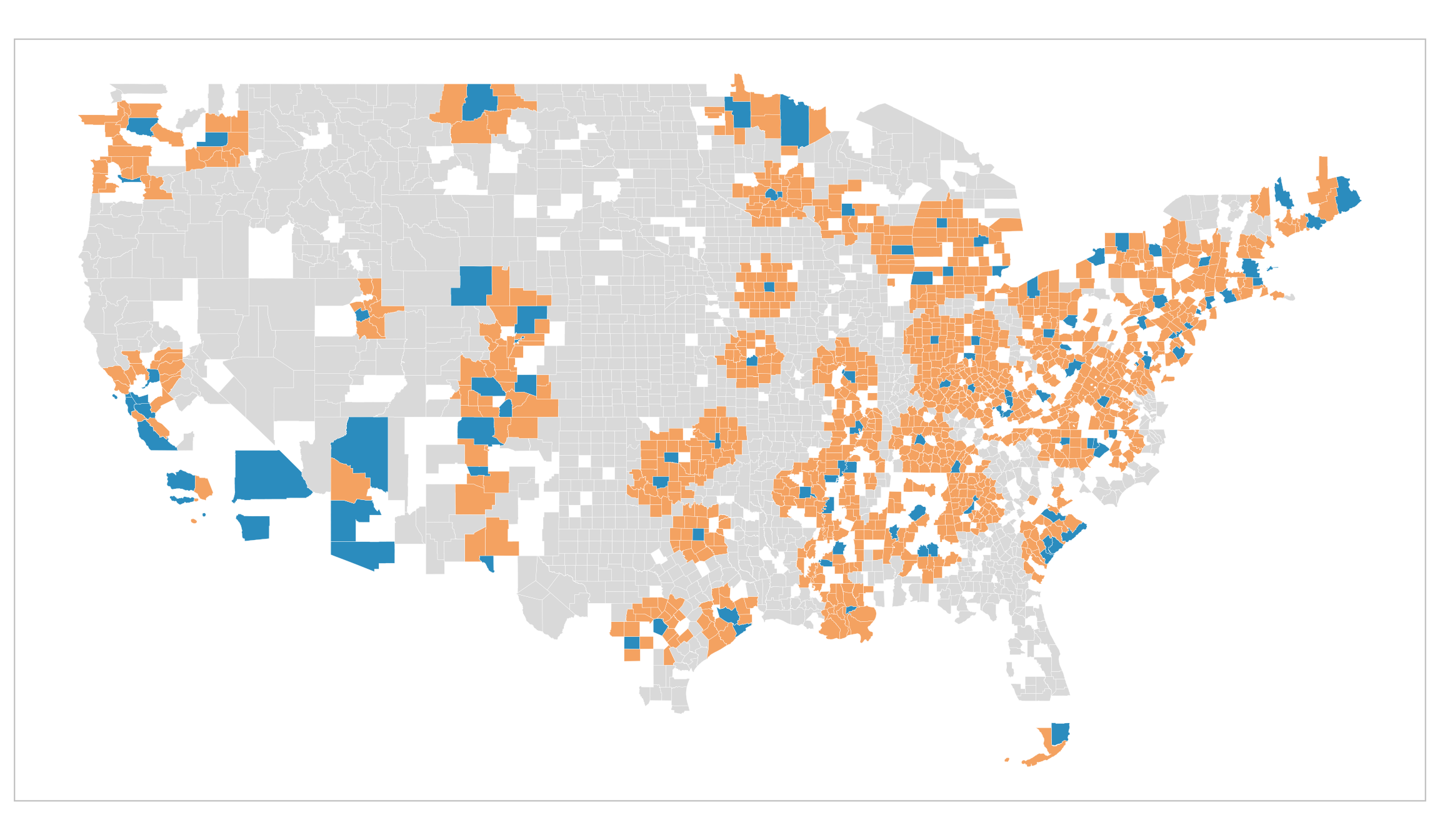}
    {\footnotesize (c) 1975}
  \end{minipage}
  \hfill
  \begin{minipage}{0.48\textwidth}
    \centering
    \includegraphics[width=\linewidth]{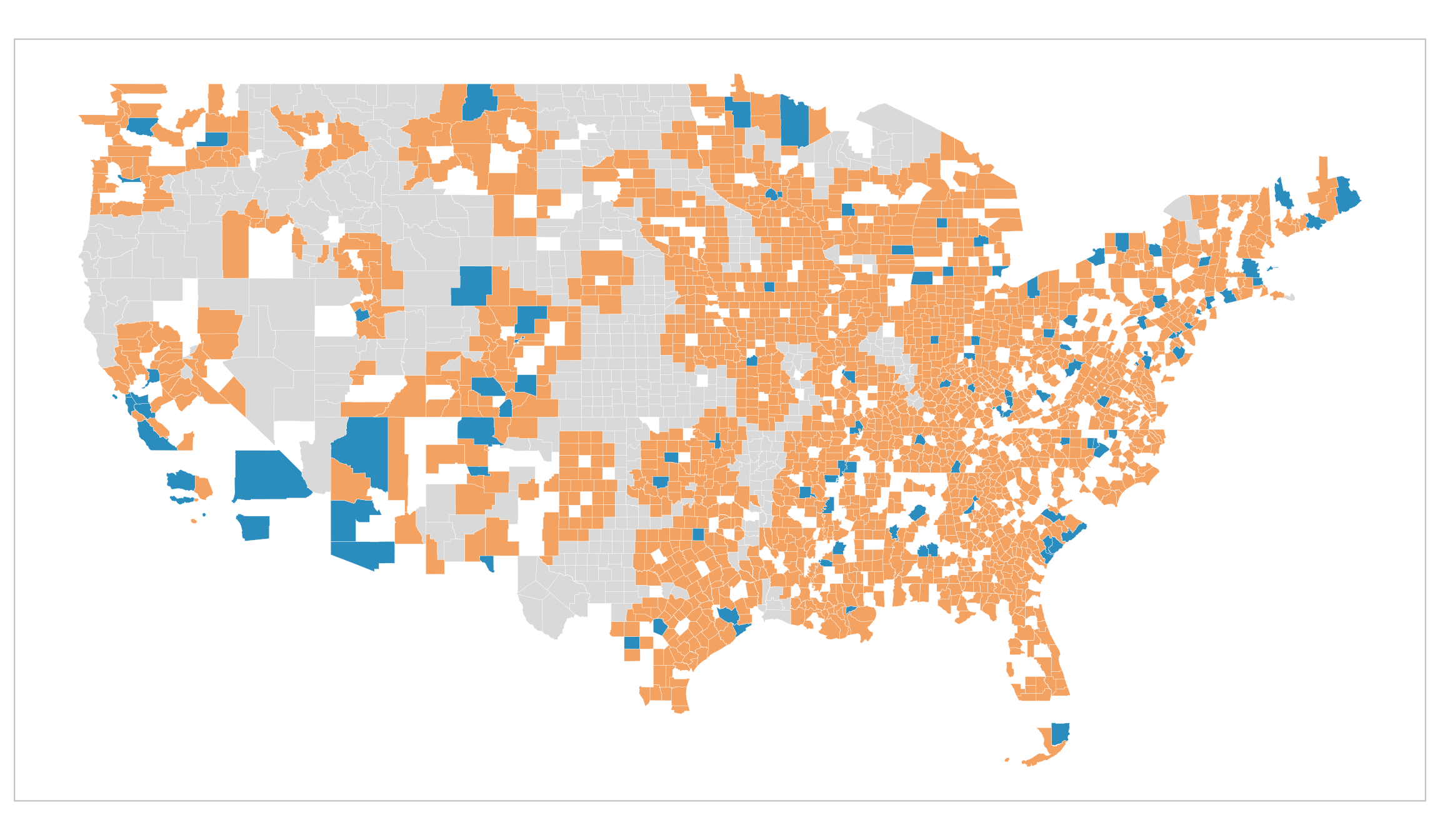}
    {\footnotesize (d) 1978}
  \end{minipage}
  \caption{Exposure under the Community Health Center rollout}
  \label{fig:application_spatial_spillover_exposure}
  \begin{minipage}{0.92\linewidth}
    \footnotesize Notes: Gray counties are unexposed and untreated, orange counties are exposed but untreated, and blue counties are treated. A county is exposed if at least one other county within 75 miles has adopted. The panels classify counties by realized adoption and exposure status in the indicated calendar year.
  \end{minipage}
\end{figure}

\subsection{Empirical estimation}

For each $(g,l)$, let $t=g+l$ and set $t_b=1959$ for both primitive comparisons.
The counterfactual sets every county's adoption time to $\infty$, so $H^0_{it}=0$.
The comparison group for the total effect consists of counties with $G_i=\infty$ and $H^1_{it}=0$.
For covariate adjustment, $X_i$ contains four indicators for the five 1960 urban categories and ten continuous baseline characteristics used by \citet{bailey2015war}: urban share; rural or farm share; shares aged 0--4 and 65 or older; nonwhite share; shares with at least 12 years and fewer than four years of schooling; shares of families with 1959 income below 3,000 dollars and above 10,000 dollars; and active physicians per 1,000 residents in 1960.
The schooling share for Bonneville County is corrected from 5908 to 59.08.
Restricting the sample to counties with complete baseline data excludes 17 counties, all of which are never-treated counties. Every county in the estimation sample has a finite physician rate, so no rate is imputed.
For each cohort and event time, we estimate the outcome regression for the total effect at $h=0$ by least squares with population weights. The regressors are an intercept and the 14 covariates in $X_i$.
The version of Equation~\eqref{eq:te_or_estimator} with analysis weights, defined in Section~\ref{sec:analysis-weights}, then standardizes the fitted trend among never-treated counties that are unexposed to the target cohort.

For switching, we estimate the outcome regression among never-treated counties by least squares with population weights.
The regressors include a saturated set of indicators for $H^1_{it}$ and the 14 covariates in $X_i$.
The indicators for exposure states supply the intercepts, while the covariate coefficients are common across exposure states.
The version of Equation~\eqref{eq:sw_or_estimator} with analysis weights standardizes the fitted trend among never-treated counties to the same treated cohort and population weights used for the total effect.
The spillover estimate is the total estimate minus the switching estimate for the same nine cohorts at event times $l = 0$ through $l = 4$.
For the spillover effect among never-treated counties, we evaluate the regression fitted among unexposed counties for each of these counties and apply the weighted analogue of \eqref{eq:never_treated_spillover_or}.
The estimate for each event time averages the resulting effects across calendar times with the same cohort weights.

We estimate the total and switching effects, the spillover effect among never-treated counties, the regression coefficients, and the cohort masses jointly.
The joint covariance uses a Wendland $C^2$ spatial kernel with a bandwidth of 100 miles.
The delta method for the effects by event time includes estimation of the weights used to aggregate across cohorts.
The Standard DID benchmark is estimated separately.
For each cohort and event time, it compares changes from $g-1$ to $g+l$, weighted by population, between cohort $g$ and all never-treated counties, without exposure adjustment. It then uses the same cohort weights to aggregate by event time.
The Standard DID benchmark estimates the contrast in Section~\ref{sec:conventional-did} with population weights.
Under interference, that contrast generally differs from both the total and switching effects.

\subsection{Pre-tests}

We apply the diagnostics in Section~\ref{sec:pre-treatment-diagnostics}.
No county is exposed under the 75-mile mapping in 1960--1964, so $H^1_{is}=H^0_{is}=0$ in these years.
For the Total effect, we estimate $P^{TE}(g,s)$ for each of the nine cohorts and $s=1960,\ldots,1964$, then aggregate the cohort estimates using the same analysis weights as the estimates after adoption.
For the Switching effect, we estimate $P^{SW}(g,g+k)$ at $k=-5,\ldots,-1$ using the realized exposure state and the matched comparison among never-treated counties.
Both diagnostics use the same covariates, population weights, and 100-mile spatial HAC rule as the main estimates.

Figure~\ref{fig:application_pretrend_v2} reports the estimates and pointwise 95 percent confidence intervals.
Every pointwise interval includes zero for both diagnostics.
The joint tests do not reject at the 5 percent level for either the Total effect ($p=0.308$) or the Switching effect ($p=0.303$).
The joint tests assess the plausibility of the corresponding parallel trends assumptions.

\begin{figure}[h]
  \centering
  \begin{subfigure}[t]{0.485\textwidth}
    \centering
    \includegraphics[width=\linewidth]{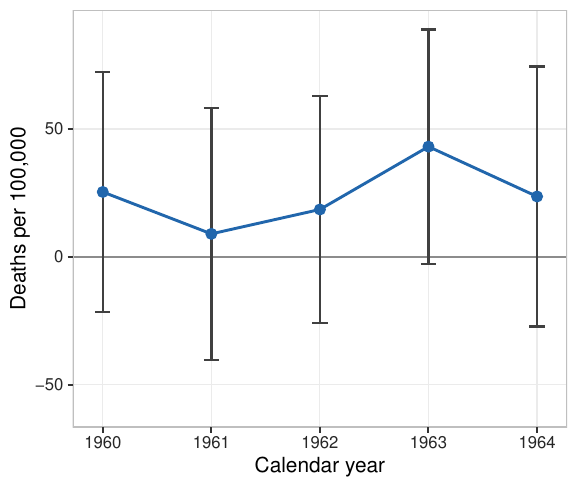}
    \caption{Parallel trends for the Total effect}
    \label{fig:application_pretrend_total_v2}
  \end{subfigure}
  \hfill
  \begin{subfigure}[t]{0.485\textwidth}
    \centering
    \includegraphics[width=\linewidth]{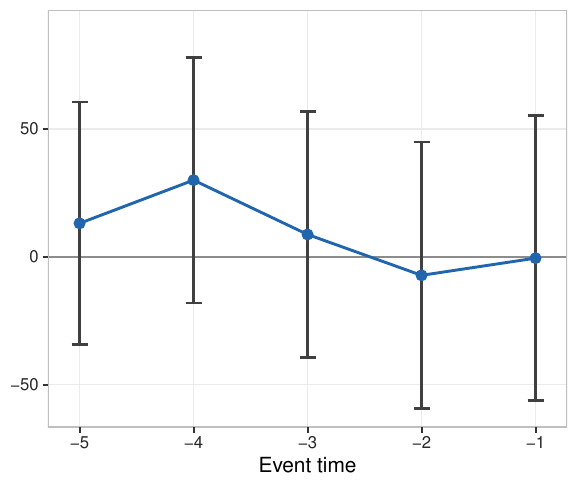}
    \caption{Parallel trends for the Switching effect}
    \label{fig:application_pretrend_switching_v2}
  \end{subfigure}
  \caption{Pre-treatment diagnostics for the Community Health Center rollout}
  \label{fig:application_pretrend_v2}
  \begin{minipage}{0.94\linewidth}
    \footnotesize Notes: Panel (a) aggregates the Total diagnostic by calendar year over the nine cohorts. Panel (b) aggregates the Switching diagnostic by event time over the same cohorts. Points are estimates, and vertical intervals are pointwise 95 percent confidence intervals from the 100-mile spatial HAC covariance. The separate joint Wald tests use all five estimates in the corresponding panel and the full covariance among them.
  \end{minipage}
\end{figure}

\subsection{Results}

Figure~\ref{fig:application_spatial_spillover_exposure} shows that exposure becomes common as the rollout proceeds.

Figure~\ref{fig:application_event_time_v2} and Table~\ref{tab:application_did_comparison} report negative total effect estimates at event times $l = 0$ through $l = 4$.
The estimates range from $-67.1$ at $l = 0$ to $-232.8$ at $l = 4$.
The pointwise 95 percent interval includes zero at $l = 0$ and excludes zero from $l = 1$ onward.

The switching estimates are also negative at every event time.
Their pointwise intervals include zero through $l = 2$ and exclude zero at $l = 3$ and $l = 4$.
The spillover estimates equal $-49.4$, $-64.4$, $-89.9$, $-113.1$, and $-159.1$.
Every pointwise 95 percent interval for the spillover effect excludes zero.
The spillover estimates are the differences between the total and switching estimates under the exposure mapping with a 75-mile radius.
The corresponding estimand compares realized exposure with zero exposure under no own adoption.
At event times $l = 0$ through $l = 4$, the estimated spillover effects among never-treated counties are $-13.9$, $-18.5$, $-25.3$, $-33.1$, and $-45.7$, respectively.
The pointwise 95 percent intervals exclude zero at $l = 3$ and $l = 4$.
Under the mapping with a 75-mile radius, these effects compare realized exposure with zero exposure among counties that never adopt.
The Standard DID estimates are also negative at event times $l = 0$ through $l = 4$.
\citet{bailey2015war} report an estimate of $-41.1$ with a standard error of $9.6$ for event years 0 through 4 in their baseline specification (their Table 2, Panel B, column 2).
Appendix Figure~\ref{fig:application_mapping_sensitivity_v2} reports the estimates under alternative exposure mappings.

\begin{table}[!htbp]
  \centering
  \caption{Estimated effects of Community Health Centers on mortality}
  \label{tab:application_did_comparison}
  \footnotesize
\setlength{\tabcolsep}{4pt}
\begin{threeparttable}
\begin{tabular}{lccccc}
\toprule
 & & & \multicolumn{2}{c}{Spillover} & \\
\cmidrule(lr){4-5}
Event time & Total & Switching & Treated counties & Never-treated counties & Standard DID \\
\midrule
$l = 0$ & \shortstack[c]{-67.1\\{\scriptsize [-147.3, 13.0]}} & \shortstack[c]{-17.8\\{\scriptsize [-77.0, 41.4]}} & \shortstack[c]{-49.4\\{\scriptsize [-95.2, -3.6]}} & \shortstack[c]{-13.9\\{\scriptsize [-30.2, 2.5]}} & \shortstack[c]{-24.1\\{\scriptsize [-48.1, -0.1]}} \\
$l = 1$ & \shortstack[c]{-108.5\\{\scriptsize [-188.7, -28.3]}} & \shortstack[c]{-44.1\\{\scriptsize [-98.7, 10.5]}} & \shortstack[c]{-64.4\\{\scriptsize [-117.8, -11.0]}} & \shortstack[c]{-18.5\\{\scriptsize [-39.1, 2.2]}} & \shortstack[c]{-36.9\\{\scriptsize [-59.9, -14.0]}} \\
$l = 2$ & \shortstack[c]{-140.0\\{\scriptsize [-225.9, -54.1]}} & \shortstack[c]{-50.2\\{\scriptsize [-103.9, 3.5]}} & \shortstack[c]{-89.9\\{\scriptsize [-153.8, -25.9]}} & \shortstack[c]{-25.3\\{\scriptsize [-51.1, 0.5]}} & \shortstack[c]{-62.7\\{\scriptsize [-87.4, -38.0]}} \\
$l = 3$ & \shortstack[c]{-170.0\\{\scriptsize [-254.1, -86.0]}} & \shortstack[c]{-56.9\\{\scriptsize [-107.5, -6.3]}} & \shortstack[c]{-113.1\\{\scriptsize [-185.3, -40.9]}} & \shortstack[c]{-33.1\\{\scriptsize [-61.9, -4.4]}} & \shortstack[c]{-64.1\\{\scriptsize [-89.5, -38.7]}} \\
$l = 4$ & \shortstack[c]{-232.8\\{\scriptsize [-324.7, -140.9]}} & \shortstack[c]{-73.7\\{\scriptsize [-124.0, -23.5]}} & \shortstack[c]{-159.1\\{\scriptsize [-246.9, -71.3]}} & \shortstack[c]{-45.7\\{\scriptsize [-81.3, -10.2]}} & \shortstack[c]{-74.9\\{\scriptsize [-99.8, -50.1]}} \\
\bottomrule
\end{tabular}
\begin{tablenotes}[flushleft]
\footnotesize
\item Notes: The outcome is the mortality rate from all causes among adults aged 50 and older, adjusted for age and expressed per 100,000 residents. Entries are point estimates with pointwise 95 percent confidence intervals in brackets. The intervals for Total, Switching, and both Spillover effects use their joint spatial HAC covariance with a 100-mile Wendland $C^2$ kernel. Standard DID is estimated separately with the same kernel and bandwidth.
\end{tablenotes}
\end{threeparttable}

\end{table}

\begin{figure}[h]
  \centering
  \includegraphics[width=0.95\textwidth]{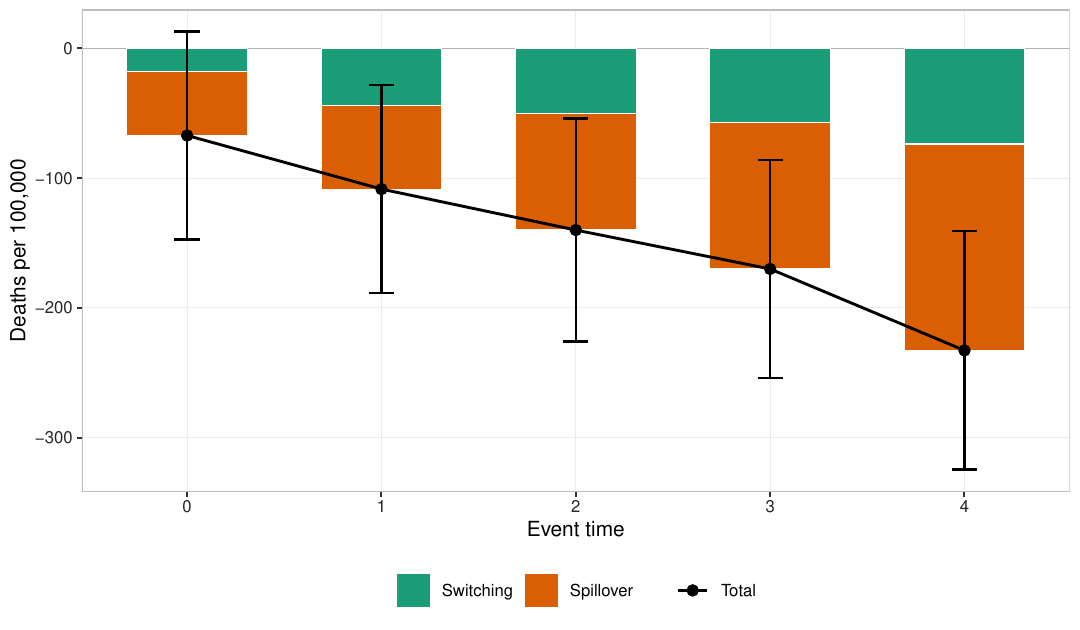}
  \caption{Decomposition of Community Health Center mortality effects by event time}
  \label{fig:application_event_time_v2}
  \begin{minipage}{0.92\linewidth}
    \footnotesize Notes: The black line and vertical intervals report the Total effect and its pointwise 95 percent confidence interval from the jointly estimated spatial HAC covariance. The green and orange segments report the Switching and Spillover point estimates, respectively, where Spillover equals Total minus Switching. All three series average the same nine cohorts with the same population weights. Pointwise intervals for Switching and Spillover appear in Table~\ref{tab:application_did_comparison}. The outcome is the mortality rate from all causes among adults aged 50 and older, adjusted for age and expressed per 100,000 residents.
  \end{minipage}
\end{figure}

\section{Conclusion}\label{sec:conclusion}

This paper defines total, switching, and spillover effects from realized and counterfactual exposure levels under an exposure mapping specified by the researcher.
Conditional parallel trends at the counterfactual exposure $H^0_{it}$ identifies the total effect directly.
The corresponding comparison at the realized exposure $H^1_{it}$ identifies the switching effect, and their difference identifies the spillover effect for the same target population.
The comparison at $H^0_{it}$ identifies the spillover effect among units that never adopt relative to the counterfactual exposure for the cohort.
Parallel trends across exposure states and transportability provide an alternative route to the spillover effect for the target cohort under no own adoption.

The regression and DR estimators treat the total and switching effects for cohorts that adopt as primitive and obtain their spillover effect by subtraction.
The spillover effect among never-treated counties is estimated directly by evaluating the fitted outcome regression for these counties at their counterfactual exposure.
The joint representation retains the covariance between the total and switching estimators required for inference on their spillover difference.
The Monte Carlo results show small bias under correct regression specification and substantial bias reduction from the DR correction when the switching outcome regression is misspecified but the odds are correctly specified.
RMSE generally decreases with sample size, and pointwise coverage is closer to 0.95 at $N=500$ in most designs.

In the empirical application, the estimated total effect on mortality is negative at event times $l = 0$ through $l = 4$.
The pointwise intervals for the total effect exclude zero from $l = 1$ onward.
The spillover estimates are also negative, and their pointwise intervals exclude zero at every event time.
The spillover estimates among never-treated counties are negative, and their pointwise intervals exclude zero at $l = 3$ and $l = 4$.
The Standard DID estimates are also negative at event times $l = 0$ through $l = 4$.
\bibliographystyle{aer}
\bibliography{ref}

@article{butts2024difference,
  title   = {Difference-in-differences with spatial spillovers},
  author  = {Butts, Kyle},
  journal = {arXiv preprint arXiv:2105.03737},
  year    = {2024}
}

@techreport{lasso2026spillover,
  author      = {Lasso-Jaramillo, Daniel},
  title       = {Spillover Gridlock: Revisiting Interference in Difference-in-Differences},
  institution = {Universidad de los Andes, Facultad de Econom{\'i}a, CEDE},
  type        = {Documentos CEDE},
  number      = {2026-26},
  year        = {2026},
  month       = {May},
  url         = {https://ideas.repec.org/p/col/000089/022570.html}
}

@article{VAZQUEZBARE2023105237,
  title    = {Identification and estimation of spillover effects in randomized experiments},
  journal  = {Journal of Econometrics},
  volume   = {237},
  number   = {1},
  pages    = {105237},
  year     = {2023},
  issn     = {0304-4076},
  doi      = {https://doi.org/10.1016/j.jeconom.2021.10.014},
  url      = {https://www.sciencedirect.com/science/article/pii/S0304407621003067},
  author   = {Gonzalo Vazquez-Bare}
}

@article{borusyak2024revisiting,
  title   = {Revisiting Event-Study Designs: Robust and Efficient Estimation},
  author  = {Borusyak, Kirill and Jaravel, Xavier and Spiess, Jann},
  journal = {The Review of Economic Studies},
  volume  = {91},
  number  = {6},
  pages   = {3253--3285},
  year    = {2024},
  doi     = {10.1093/restud/rdae007},
  url     = {https://doi.org/10.1093/restud/rdae007}
}

@article{berg2021spillover,
  title   = {Spillover Effects in Empirical Corporate Finance},
  author  = {Berg, Tobias and Reisinger, Markus and Streitz, Daniel},
  journal = {Journal of Financial Economics},
  volume  = {142},
  number  = {3},
  pages   = {1109--1127},
  year    = {2021},
  doi     = {10.1016/j.jfineco.2021.04.039},
  url     = {https://doi.org/10.1016/j.jfineco.2021.04.039}
}

@article{goodman2021difference,
  title     = {Difference-in-differences with variation in treatment timing},
  author    = {Goodman-Bacon, Andrew},
  journal   = {Journal of econometrics},
  volume    = {225},
  number    = {2},
  pages     = {254--277},
  year      = {2021},
  publisher = {Elsevier}
}

@article{delgado2015difference,
  title     = {Difference-in-differences techniques for spatial data: Local autocorrelation and spatial interaction},
  author    = {Delgado, Michael S and Florax, Raymond JGM},
  journal   = {Economics Letters},
  volume    = {137},
  pages     = {123--126},
  year      = {2015},
  doi       = {10.1016/j.econlet.2015.10.035},
  url       = {https://doi.org/10.1016/j.econlet.2015.10.035},
  publisher = {Elsevier}
}

@techreport{clarke2017estimating,
  title       = {Estimating Difference-in-Differences in the Presence of Spillovers},
  author      = {Clarke, Damian},
  year        = {2017},
  type        = {MPRA Paper},
  number      = {81604},
  institution = {University Library of Munich},
  url         = {https://mpra.ub.uni-muenchen.de/81604/}
}

@article{huber2021framework,
  title     = {A framework for separating individual-level treatment effects from spillover effects},
  author    = {Huber, Martin and Steinmayr, Andreas},
  journal   = {Journal of Business \& Economic Statistics},
  volume    = {39},
  number    = {2},
  pages     = {422--436},
  year      = {2021},
  publisher = {Taylor \& Francis}
}

@article{sant2020doubly,
  title     = {Doubly robust difference-in-differences estimators},
  author    = {Sant'Anna, Pedro HC and Zhao, Jun},
  journal   = {Journal of econometrics},
  volume    = {219},
  number    = {1},
  pages     = {101--122},
  year      = {2020},
  publisher = {Elsevier}
}

@article{10.1214/16-AOAS1005,
  author    = {Peter M. Aronow and Cyrus Samii},
  title     = {{Estimating average causal effects under general interference, with application to a social network experiment}},
  volume    = {11},
  journal   = {The Annals of Applied Statistics},
  number    = {4},
  publisher = {Institute of Mathematical Statistics},
  pages     = {1912 -- 1947},
  year      = {2017},
  doi       = {10.1214/16-AOAS1005},
  url       = {https://doi.org/10.1214/16-AOAS1005}
}

@article{savje2024causal,
  title     = {Causal inference with misspecified exposure mappings: separating definitions and assumptions},
  author    = {S{\"a}vje, Fredrik},
  journal   = {Biometrika},
  volume    = {111},
  number    = {1},
  pages     = {1--15},
  year      = {2024},
  publisher = {Oxford University Press}
}

@article{leung2022causal,
  title     = {Causal inference under approximate neighborhood interference},
  author    = {Leung, Michael P},
  journal   = {Econometrica},
  volume    = {90},
  number    = {1},
  pages     = {267--293},
  year      = {2022},
  publisher = {Wiley Online Library}
}

@techreport{fiorini2025simple,
  title       = {A Simple Approach to Staggered Difference-in-Differences in the Presence of Spillovers},
  author      = {Fiorini, Mario and Lee, Wooyong and Pfeifer, Gregor},
  year        = {2024},
  institution = {CESifo},
  type        = {CESifo Working Paper},
  number      = {11011},
  note        = {March 2024}
}

@article{ROTH20232218,
  title    = {What’s trending in difference-in-differences? A synthesis of the recent econometrics literature},
  journal  = {Journal of Econometrics},
  volume   = {235},
  number   = {2},
  pages    = {2218-2244},
  year     = {2023},
  issn     = {0304-4076},
  doi      = {https://doi.org/10.1016/j.jeconom.2023.03.008},
  url      = {https://www.sciencedirect.com/science/article/pii/S0304407623001318},
  author   = {Jonathan Roth and Pedro H.C. Sant’Anna and Alyssa Bilinski and John Poe}
}

@article{callaway2021difference,
  title     = {Difference-in-differences with multiple time periods},
  author    = {Callaway, Brantly and Sant'Anna, Pedro HC},
  journal   = {Journal of econometrics},
  volume    = {225},
  number    = {2},
  pages     = {200--230},
  year      = {2021},
  publisher = {Elsevier}
}

@article{sun2021estimating,
  title     = {Estimating dynamic treatment effects in event studies with heterogeneous treatment effects},
  author    = {Sun, Liyang and Abraham, Sarah},
  journal   = {Journal of econometrics},
  volume    = {225},
  number    = {2},
  pages     = {175--199},
  year      = {2021},
  publisher = {Elsevier}
}

@article{gardner2025two,
  title   = {Two-Stage Differences in Differences},
  author  = {Gardner, John and Thakral, Neil and T{\^o}, Linh T and Yap, Luther},
  journal = {Mimeo},
  year    = {2025},
  url     = {https://neilthakral.github.io/files/papers/2sdd.pdf}
}

@article{bailey2015war,
  title     = {The War on Poverty's experiment in public medicine: Community health centers and the mortality of older Americans},
  author    = {Bailey, Martha J and Goodman-Bacon, Andrew},
  journal   = {American Economic Review},
  volume    = {105},
  number    = {3},
  pages     = {1067--1104},
  year      = {2015},
  publisher = {American Economic Association 2014 Broadway, Suite 305, Nashville, TN 37203}
}

@misc{sun2025differenceindifferencesnetworkinterference,
  title         = {Difference-in-Differences Under Network Interference},
  author        = {Kuan Sun and Zhiguo Xiao},
  year          = {2025},
  eprint        = {2509.24259},
  archiveprefix = {arXiv},
  primaryclass  = {stat.ME},
  url           = {https://arxiv.org/abs/2509.24259}
}

@article{conley1999gmm,
  author  = {Conley, Timothy G.},
  title   = {GMM Estimation with Cross Sectional Dependence},
  journal = {Journal of Econometrics},
  volume  = {92},
  number  = {1},
  pages   = {1--45},
  year    = {1999},
  doi     = {10.1016/S0304-4076(98)00084-0},
  url     = {https://www.sciencedirect.com/science/article/pii/S0304407698000840}
}

@incollection{newey1994large,
  author    = {Newey, Whitney K. and McFadden, Daniel},
  title     = {Large Sample Estimation and Hypothesis Testing},
  booktitle = {Handbook of Econometrics},
  volume    = {4},
  chapter   = {36},
  pages     = {2111--2245},
  year      = {1994},
  publisher = {Elsevier},
  doi       = {10.1016/S1573-4412(05)80005-4},
  url       = {https://www.sciencedirect.com/science/article/pii/S1573441205800054}
}

@article{KOJEVNIKOV2021882,
  title    = {Limit theorems for network dependent random variables},
  journal  = {Journal of Econometrics},
  volume   = {222},
  number   = {2},
  pages    = {882-908},
  year     = {2021},
  issn     = {0304-4076},
  doi      = {https://doi.org/10.1016/j.jeconom.2020.05.019},
  url      = {https://www.sciencedirect.com/science/article/pii/S0304407620302402},
  author   = {Denis Kojevnikov and Vadim Marmer and Kyungchul Song}
}

@article{wendland1995piecewise,
  author  = {Wendland, Holger},
  title   = {Piecewise Polynomial, Positive Definite and Compactly Supported Radial Functions of Minimal Degree},
  journal = {Advances in Computational Mathematics},
  volume  = {4},
  pages   = {389--396},
  year    = {1995},
  doi     = {10.1007/BF02123482},
  url     = {https://doi.org/10.1007/BF02123482}
}

@article{10.1257/pandp.20261108,
  author  = {Xu, Ruonan},
  title   = {Dynamic Difference-in-Differences with Interference},
  journal = {AEA Papers and Proceedings},
  volume  = {116},
  year    = {2026},
  month   = {May},
  pages   = {58-63},
  doi     = {10.1257/pandp.20261108},
  url     = {https://www.aeaweb.org/articles?id=10.1257/pandp.20261108}
}

@article{XU2026106304,
title = {Difference-in-differences with interference},
journal = {Journal of Econometrics},
volume = {257},
pages = {106304},
year = {2026},
issn = {0304-4076},
doi = {10.1016/j.jeconom.2026.106304},
url = {https://www.sciencedirect.com/science/article/pii/S0304407626001259},
author = {Ruonan Xu}
}

\newpage

\appendix

\section{Identification proofs}\label{app:identification-proofs}

This appendix proves the identification results in Section~\ref{sec:identification}.

To shorten the derivations, define the observable conditional mean among units that never adopt
\begin{align*}
  \mu_t(x,h) & \equiv \mathbb E[Y_{it} - Y_{i,t_b}\mid G_i = \infty,\ X_i = x,\ H^1_{it} = h]                                                                                 \\
             & = \mathbb E[Y_{it}(\mathbf G) - Y_{i,t_b}(\mathbf G)\mid G_i = \infty,\ X_i = x,\ H^1_{it} = h] \tag*{\text{(by Assumption~\ref{assumption:consistency})}}     \\
             & = \mathbb E[Y_{it}(\infty,h) - Y_{i,t_b}(\infty,0)\mid G_i = \infty,\ X_i = x,\ H^1_{it} = h]. \tag*{\text{(by Assumption~\ref{assumption:exposure_mapping})}}
\end{align*}

\begin{proof}[Proof of Proposition~\ref{prop:dte}]
  Fix $(g,l)$ and set $t=g+l$.
  \begin{align*}
    \tau^{TE}(g,l) & = \mathbb E_g\!\left[Y_{it}(g,H^1_{it}) - Y_{it}(\infty,H^0_{it})\right]                                                                                                                                 \\
                   & = \mathbb E_g\!\left[Y_{it}(g,H^1_{it}) - Y_{i,t_b}(g,0)\right] - \mathbb E_g\!\left[Y_{it}(\infty,H^0_{it}) - Y_{i,t_b}(g,0)\right]                                                                     \\
                   & = \mathbb E_g\!\left[Y_{it}(g,H^1_{it}) - Y_{i,t_b}(g,0)\right] - \mathbb E_g\!\left[Y_{it}(\infty,H^0_{it}) - Y_{i,t_b}(\infty,0)\right] \tag*{\text{(by Assumption~\ref{assumption:no_anticipation})}} \\
                   & = \mathbb E_g[Y_{it}(\mathbf G) - Y_{i,t_b}(\mathbf G)] - \mathbb E_g\!\left[Y_{it}(\infty,H^0_{it}) - Y_{i,t_b}(\infty,0)\right] \tag*{\text{(by Assumption~\ref{assumption:exposure_mapping})}}        \\
                   & = \mathbb E_g[Y_{it} - Y_{i,t_b}] - \mathbb E_g\!\left[Y_{it}(\infty,H^0_{it}) - Y_{i,t_b}(\infty,0)\right] \tag*{\text{(by Assumption~\ref{assumption:consistency})}}                                   \\
                   & = \mathbb E_g[Y_{it} - Y_{i,t_b}] - \mathbb E_g[\mu_t(X_i,H^0_{it})] \tag*{\text{(by Assumption~\ref{assumption:te_parallel_trends})}}                                                                   \\
                   & = \mathbb E_g[Y_{it} - Y_{i,t_b}] - \mathbb E_g\!\left[ \sum_{h \in \mathcal H_N}\1\{H^0_{it} = h\} \mathbb E[Y_{it} - Y_{i,t_b}\mid G_i = \infty,\ X_i,\ H^1_{it} = h] \right].
  \end{align*}
  Assumption~\ref{assumption:overlap} ensures that $\mu_t(X_i,H^0_{it})$ is defined on the support over which the final expectation is taken.
\end{proof}

\begin{proof}[Proof of Proposition~\ref{prop:dse}]
  Fix \((g,l)\) and set \(t=g+l\).
  \begin{align*}
    \tau^{SW}(g,l) & = \mathbb E_g\!\left[Y_{it}(g,H^1_{it}) - Y_{it}(\infty,H^1_{it})\right]                                                                                                                                 \\
                   & = \mathbb E_g\!\left[Y_{it}(g,H^1_{it}) - Y_{i,t_b}(g,0)\right] - \mathbb E_g\!\left[Y_{it}(\infty,H^1_{it}) - Y_{i,t_b}(g,0)\right]                                                                     \\
                   & = \mathbb E_g\!\left[Y_{it}(g,H^1_{it}) - Y_{i,t_b}(g,0)\right] - \mathbb E_g\!\left[Y_{it}(\infty,H^1_{it}) - Y_{i,t_b}(\infty,0)\right] \tag*{\text{(by Assumption~\ref{assumption:no_anticipation})}} \\
                   & = \mathbb E_g[Y_{it}(\mathbf G) - Y_{i,t_b}(\mathbf G)] - \mathbb E_g\!\left[Y_{it}(\infty,H^1_{it}) - Y_{i,t_b}(\infty,0)\right] \tag*{\text{(by Assumption~\ref{assumption:exposure_mapping})}}        \\
                   & = \mathbb E_g[Y_{it} - Y_{i,t_b}] - \mathbb E_g\!\left[Y_{it}(\infty,H^1_{it}) - Y_{i,t_b}(\infty,0)\right] \tag*{\text{(by Assumption~\ref{assumption:consistency})}}                                   \\
                   & = \mathbb E_g[Y_{it} - Y_{i,t_b}] - \mathbb E_g[\mu_t(X_i,H^1_{it})] \tag*{\text{(by Assumption~\ref{assumption:dse_parallel_trends})}}                                                                  \\
                   & = \mathbb E_g[Y_{it} - Y_{i,t_b}] - \mathbb E_g\!\left[ \sum_{h \in \mathcal H_N}\1\{H^1_{it} = h\} \mathbb E[Y_{it} - Y_{i,t_b}\mid G_i = \infty,\ X_i,\ H^1_{it} = h] \right].
  \end{align*}
  Assumption~\ref{assumption:overlap} ensures that $\mu_t(X_i,H^1_{it})$ is defined on the support over which the final expectation is taken.
\end{proof}

\begin{proof}[Proof of Proposition~\ref{prop:never_treated_spillover}]
  Fix $(g,l)$ and set $t=g+l$.
  \begin{align*}
    \tau_\infty^{SP}(g,l) & = \mathbb E_\infty\!\left[Y_{it}(\infty,H^1_{it}) - Y_{it}(\infty,H^0_{it})\right]                                                                                                                          \\
                          & = \mathbb E_\infty\!\left[Y_{it}(\infty,H^1_{it}) - Y_{i,t_b}(\infty,0)\right] - \mathbb E_\infty\!\left[Y_{it}(\infty,H^0_{it}) - Y_{i,t_b}(\infty,0)\right]                                               \\
                          & = \mathbb E_\infty[Y_{it}(\mathbf G) - Y_{i,t_b}(\mathbf G)] - \mathbb E_\infty\!\left[Y_{it}(\infty,H^0_{it}) - Y_{i,t_b}(\infty,0)\right] \tag*{\text{(by Assumption~\ref{assumption:exposure_mapping})}} \\
                          & = \mathbb E_\infty[Y_{it} - Y_{i,t_b}] - \mathbb E_\infty\!\left[Y_{it}(\infty,H^0_{it}) - Y_{i,t_b}(\infty,0)\right] \tag*{\text{(by Assumption~\ref{assumption:consistency})}}                            \\
                          & = \mathbb E_\infty[Y_{it} - Y_{i,t_b}] - \mathbb E_\infty[\mu_t(X_i,H^0_{it})] \tag*{\text{(by Assumption~\ref{assumption:te_parallel_trends})}}                                                            \\
                          & = \mathbb E_\infty[Y_{it} - Y_{i,t_b}] - \mathbb E_\infty\!\left[ \sum_{h \in \mathcal H_N}\1\{H^0_{it} = h\} \mathbb E[Y_{it} - Y_{i,t_b}\mid G_i = \infty,\ X_i,\ H^1_{it} = h] \right].
  \end{align*}
  Assumption~\ref{assumption:overlap} ensures that $\mu_t(X_i,H^0_{it})$ is defined on the support among units that never adopt over which the final expectation is taken.
\end{proof}

\begin{proof}[Proof of Proposition~\ref{prop:alternative-identification}]
  Fix $(g,l)$ and set $t=g+l$.
  For every $(x,h^1,h^0)$ in the target support, the representation of $\mu_t$ above gives
  \begin{align*}
    \mu_t(x,h^1) - \mu_t(x,h^0) & = \mathbb E\!\left[Y_{it}(\infty,h^1) - Y_{i,t_b}(\infty,0) \mid G_i = \infty,\ X_i = x,\ H^1_{it} = h^1\right]                                                                                            \\
                                & \mathrel{\phantom{=}} - \mathbb E\!\left[Y_{it}(\infty,h^0) - Y_{i,t_b}(\infty,0) \mid G_i = \infty,\ X_i = x,\ H^1_{it} = h^0\right]                                                                      \\
                                & = \mathbb E\!\left[Y_{it}(\infty,h^1) - Y_{it}(\infty,0) \mid G_i = \infty,\ X_i = x,\ H^1_{it} = h^1\right]                                                                                               \\
                                & \mathrel{\phantom{=}} - \mathbb E\!\left[Y_{it}(\infty,h^0) - Y_{it}(\infty,0) \mid G_i = \infty,\ X_i = x,\ H^1_{it} = h^0\right]                                                                         \\
                                & \mathrel{\phantom{=}} + \mathbb E\!\left[Y_{it}(\infty,0) - Y_{i,t_b}(\infty,0) \mid G_i = \infty,\ X_i = x,\ H^1_{it} = h^1\right]                                                                        \\
                                & \mathrel{\phantom{=}} - \mathbb E\!\left[Y_{it}(\infty,0) - Y_{i,t_b}(\infty,0) \mid G_i = \infty,\ X_i = x,\ H^1_{it} = h^0\right]                                                                        \\
                                & = \mathbb E\!\left[Y_{it}(\infty,h^1) - Y_{it}(\infty,0) \mid G_i = \infty,\ X_i = x,\ H^1_{it} = h^1\right]                                                                                               \\
                                & \mathrel{\phantom{=}} - \mathbb E\!\left[Y_{it}(\infty,h^0) - Y_{it}(\infty,0) \mid G_i = \infty,\ X_i = x,\ H^1_{it} = h^0\right] \tag*{\text{(by Assumption~\ref{assumption:exposure_parallel_trends})}} \\
                                & = \mathbb E\!\left[Y_{it}(\infty,h^1) - Y_{it}(\infty,h^0) \mid G_i = g,\ X_i = x,\ H^1_{it} = h^1,\ H^0_{it} = h^0\right]. \tag*{\text{(by Assumption~\ref{assumption:exposure_transportability})}}
  \end{align*}
  Averaging the conditional equality over the target cohort gives
  \begin{align*}
    \tau^{SP}(g,l) & = \mathbb E_g[\mu_t(X_i,H^1_{it}) - \mu_t(X_i,H^0_{it})]                                                                                                                \\
                   & = \mathbb E_g\!\left[ \sum_{h \in \mathcal H_N}\{\1\{H^1_{it} = h\} - \1\{H^0_{it} = h\}\} \mathbb E[Y_{it} - Y_{i,t_b}\mid G_i = \infty,\ X_i,\ H^1_{it} = h] \right].
  \end{align*}
  Assumption~\ref{assumption:overlap} ensures that both conditional means are defined on the target support.
  If the conditions of Proposition~\ref{prop:dse} also hold, then
  \begin{align*}
    \hspace{-10mm}\tau^{TE}(g,l) & = \tau^{SW}(g,l) + \tau^{SP}(g,l)                                                                                                                                                                                                     \\
                                 & = \mathbb E_g[Y_{it} - Y_{i,t_b}] - \mathbb E_g\!\left[\sum_{h \in \mathcal H_N}\1\{H^1_{it} = h\}\mu_t(X_i,h)\right]  + \mathbb E_g\!\left[ \sum_{h \in \mathcal H_N}\{\1\{H^1_{it} = h\} - \1\{H^0_{it} = h\}\}\mu_t(X_i,h) \right] \\
                                 & = \mathbb E_g[Y_{it} - Y_{i,t_b}] - \mathbb E_g\!\left[\sum_{h \in \mathcal H_N}\1\{H^0_{it} = h\}\mu_t(X_i,h)\right].
  \end{align*}
  Substituting the definition of $\mu_t$ gives the formula for the total effect.
\end{proof}

\begin{proof}[Proof of Proposition~\ref{prop:did_decomposition}]
  Fix $(g,l)$ and set $t=g+l$.
  \begin{align*}
    \tau^{DID}(g,l) & = \mathbb E_g[Y_{it} - Y_{i,g-1}] - \mathbb E_\infty[Y_{it} - Y_{i,g-1}]                                                                                                                                                           \\
                    & = \mathbb E_g[Y_{it}(\mathbf G) - Y_{i,g-1}(\mathbf G)] - \mathbb E_\infty[Y_{it}(\mathbf G) - Y_{i,g-1}(\mathbf G)] \tag*{\text{(by Assumption~\ref{assumption:consistency})}}                                                    \\
                    & = \mathbb E_g\!\left[Y_{it}(g,H^1_{it}) - Y_{i,g-1}(g,H^1_{i,g-1})\right] - \mathbb E_\infty\!\left[Y_{it}(\infty,H^1_{it}) - Y_{i,g-1}(\infty,H^1_{i,g-1})\right] \tag*{\text{(by Assumption~\ref{assumption:exposure_mapping})}} \\
                    & = \mathbb E_g\!\left[Y_{it}(g,H^1_{it}) - Y_{i,g-1}(\infty,H^1_{i,g-1})\right]  - \mathbb E_\infty\!\left[Y_{it}(\infty,H^1_{it}) - Y_{i,g-1}(\infty,H^1_{i,g-1})\right]. \tag*{\text{(by no anticipation at $g-1$)}}
  \end{align*}
  Adding and subtracting $Y_{it}(\infty,H^0_{it})$ and $Y_{i,g-1}(\infty,H^0_{i,g-1})$ within the two expectations gives
  \begin{align*}
    \tau^{DID}(g,l) & = \mathbb E_g\!\left[Y_{it}(g,H^1_{it}) - Y_{it}(\infty,H^0_{it})\right]  - \mathbb E_g\!\left[Y_{i,g-1}(\infty,H^1_{i,g-1}) - Y_{i,g-1}(\infty,H^0_{i,g-1})\right]                                                     \\
                    & \mathrel{\phantom{=}} - \bigg\{ \mathbb E_\infty\!\left[Y_{it}(\infty,H^1_{it}) - Y_{it}(\infty,H^0_{it})\right]  - \mathbb E_\infty\!\left[Y_{i,g-1}(\infty,H^1_{i,g-1}) - Y_{i,g-1}(\infty,H^0_{i,g-1})\right]\bigg\} \\
                    & \mathrel{\phantom{=}} + \bigg\{\mathbb E_g\!\left[Y_{it}(\infty,H^0_{it}) - Y_{i,g-1}(\infty,H^0_{i,g-1})\right] - \mathbb E_\infty\!\left[Y_{it}(\infty,H^0_{it}) - Y_{i,g-1}(\infty,H^0_{i,g-1})\right]\bigg\}.
  \end{align*}
\end{proof}

\section{Estimation details}\label{app:estimation-details}

This appendix gives the outcome specifications, the robustness argument, and the joint estimating equations.

\subsection{Population functionals and nuisance specifications}\label{app:parametric-first-stages}

For a predetermined analysis weight $\varpi_i$, define
\begin{align*}
  m_{\varpi,t}(x,h)              & = \dfrac{\mathbb E[\varpi_i\1\{G_i=\infty\}\Delta Y_{it}\mid X_i=x,\ H^1_{it}=h]}{\mathbb E[\varpi_i\1\{G_i=\infty\}\mid X_i=x,\ H^1_{it}=h]}, \\
  \lambda_{\varpi,TE}^{g,l}(x,h) & = \dfrac{\mathbb E[\varpi_i\1\{G_i=g,\ H^0_{it}=h\}\mid X_i=x]}{\mathbb E[\varpi_i\1\{G_i=\infty,\ H^1_{it}=h\}\mid X_i=x]},                   \\
  \lambda_{\varpi,SW}^{g,l}(x,h) & = \dfrac{\mathbb E[\varpi_i\1\{G_i=g,\ H^1_{it}=h\}\mid X_i=x]}{\mathbb E[\varpi_i\1\{G_i=\infty,\ H^1_{it}=h\}\mid X_i=x]}.
\end{align*}
Let $\widetilde m_{\varpi,t}^{TE,g,l}$, $\widetilde m_{\varpi,t}^{SW,g,l}$, $\widetilde\lambda_{\varpi,TE}^{g,l}$, and $\widetilde\lambda_{\varpi,SW}^{g,l}$ denote the conditional mean and odds functions under the weighted specifications chosen by the researcher.
For $Q\in\{TE,SW\}$, define $\lambda_{\varpi,Q}^{g,l}$, $\widetilde\lambda_{\varpi,Q}^{g,l}$, and their fitted analogues to be zero outside the support of the numerator of $\lambda_{\varpi,Q}^{g,l}$.
The weighted doubly robust functionals are
\begin{align*}
  \tau_{\varpi,DR}^{TE}(g,l) & = \dfrac{1}{s_{\varpi,g}}\mathbb E\!\left[\varpi_i\1\{G_i=g\}\{\Delta Y_{it}-\widetilde m_{\varpi,t}^{TE,g,l}(X_i,H^0_{it})\}\right]                                                                            \\
                             & \mathrel{\phantom{=}}-\dfrac{1}{s_{\varpi,g}}\mathbb E\!\left[\varpi_i\1\{G_i=\infty\}\widetilde\lambda_{\varpi,TE}^{g,l}(X_i,H^1_{it})\{\Delta Y_{it}-\widetilde m_{\varpi,t}^{TE,g,l}(X_i,H^1_{it})\}\right], \\
  \tau_{\varpi,DR}^{SW}(g,l) & = \dfrac{1}{s_{\varpi,g}}\mathbb E\!\left[\varpi_i\1\{G_i=g\}\{\Delta Y_{it}-\widetilde m_{\varpi,t}^{SW,g,l}(X_i,H^1_{it})\}\right]                                                                            \\
                             & \mathrel{\phantom{=}}-\dfrac{1}{s_{\varpi,g}}\mathbb E\!\left[\varpi_i\1\{G_i=\infty\}\widetilde\lambda_{\varpi,SW}^{g,l}(X_i,H^1_{it})\{\Delta Y_{it}-\widetilde m_{\varpi,t}^{SW,g,l}(X_i,H^1_{it})\}\right], \\
  \tau_{\varpi,DR}^{SP}(g,l) & = \tau_{\varpi,DR}^{TE}(g,l)-\tau_{\varpi,DR}^{SW}(g,l).
\end{align*}
Setting $\varpi_i=1$ yields the functionals in Proposition~\ref{prop:dr_components}.
Under the weighted identification conditions, the spillover effect among units that never adopt has the following representation in terms of observed outcomes:
\begin{align*}
  \tau_{\varpi,\infty}^{SP}(g,l) & = \dfrac{1}{s_{\varpi,\infty}}\mathbb E\!\left[\varpi_i\1\{G_i=\infty\}\{\Delta Y_{it}-m_{\varpi,t}(X_i,H^0_{it})\}\right].
\end{align*}

For $Q\in\{TE,SW\}$, let $m_{\varpi,t}^{Q,g,l}(x,h;\beta_{g,l}^Q)$ be a conditional mean specification and let $\lambda_{\varpi,Q}^{g,l}(x,h;\alpha_{g,l}^Q)$ denote the odds induced by a propensity score model, both indexed by finite-dimensional parameters.
For the Total effect and the spillover effect among units that never adopt, the conditional mean specification is estimated using units that never adopt whose $(X_i,H^1_{it})$ values belong to the union of the supports of $(X_i,H^0_{it})\mid G_i=g$ and $(X_i,H^0_{it})\mid G_i=\infty$.
For the Switching effect, it is estimated using units that never adopt whose $(X_i,H^1_{it})$ values belong to the support of $(X_i,H^1_{it})\mid G_i=g$.
The odds specifications use pooled samples that combine units in cohort $g$ with units that never adopt.
For the Total effect, observations in cohort $g$ are evaluated at $(X_i,H^0_{it})$, and observations for units that never adopt are evaluated at $(X_i,H^1_{it})$, on the support of $(X_i,H^0_{it})\mid G_i=g$.
For the Switching effect, both groups are evaluated at $(X_i,H^1_{it})$ on the support of $(X_i,H^1_{it})\mid G_i=g$.

Let $q_{m,Q,i}^{g,l}(\beta_{g,l}^Q)$ and $q_{\lambda,Q,i}^{g,l}(\alpha_{g,l}^Q)$ denote smooth unit-level estimating functions, of the same dimensions as their parameters, formed from these samples and weighted by $\varpi_i$.
Their sample equations define $\widehat\beta_{g,l}^Q$ and $\widehat\alpha_{g,l}^Q$.
The fitted conditional mean and odds functions are obtained by evaluating the specifications at these estimates.
Weighted least-squares scores for the conditional mean and weighted logistic likelihood scores for the corresponding propensity score are leading examples.
At the population roots $\beta_{g,l}^{Q,0}$ and $\alpha_{g,l}^{Q,0}$, define
\begin{align*}
  \widetilde m_{\varpi,t}^{Q,g,l}(x,h)    & = m_{\varpi,t}^{Q,g,l}(x,h;\beta_{g,l}^{Q,0}),      \\
  \widetilde\lambda_{\varpi,Q}^{g,l}(x,h) & = \lambda_{\varpi,Q}^{g,l}(x,h;\alpha_{g,l}^{Q,0}).
\end{align*}
Correct specification of a conditional mean or odds specification means that the corresponding function at its population root equals $m_{\varpi,t}$ or $\lambda_{\varpi,Q}^{g,l}$ on the relevant support.
The specified odds and their fitted values are set to zero outside the numerator support.
With finite covariate support, saturated odds equal ratios of weighted cell frequencies.

\begin{proof}[Proof of Proposition~\ref{prop:dr_components}]
  Subtracting the weighted identification formulas from the weighted DR functionals and applying the change-of-measure identity gives
  \begin{align*}
    \tau_{\varpi,DR}^{TE}(g,l)-\tau_\varpi^{TE}(g,l) & = \dfrac{1}{\mathbb E[\varpi_i\1\{G_i=g\}]}\mathbb E\!\left[\varpi_i\1\{G_i=\infty\}\right. \\
                                                                                                               & \mathrel{\phantom{=}}\left.\times\{\widetilde\lambda_{\varpi,TE}^{g,l}(X_i,H^1_{it})-\lambda_{\varpi,TE}^{g,l}(X_i,H^1_{it})\}\{\widetilde m_{\varpi,t}^{TE,g,l}(X_i,H^1_{it})-m_{\varpi,t}(X_i,H^1_{it})\}\right], \\
    \tau_{\varpi,DR}^{SW}(g,l)-\tau_\varpi^{SW}(g,l) & = \dfrac{1}{\mathbb E[\varpi_i\1\{G_i=g\}]}\mathbb E\!\left[\varpi_i\1\{G_i=\infty\}\right. \\
                                                                                                               & \mathrel{\phantom{=}}\left.\times\{\widetilde\lambda_{\varpi,SW}^{g,l}(X_i,H^1_{it})-\lambda_{\varpi,SW}^{g,l}(X_i,H^1_{it})\}\{\widetilde m_{\varpi,t}^{SW,g,l}(X_i,H^1_{it})-m_{\varpi,t}(X_i,H^1_{it})\}\right].
  \end{align*}
  For $Q\in\{TE,SW\}$, the corresponding remainder is zero if either $\widetilde m_{\varpi,t}^{Q,g,l}=m_{\varpi,t}$ or $\widetilde\lambda_{\varpi,Q}^{g,l}=\lambda_{\varpi,Q}^{g,l}$ on the support of the numerator of $\lambda_{\varpi,Q}^{g,l}$.
  Setting $\varpi_i=1$ proves the two claims in the proposition.
  If the result holds for both $TE$ and $SW$, subtracting the two identities gives $\tau^{TE}(g,l)-\tau^{SW}(g,l)=\tau^{SP}(g,l)$.
\end{proof}

\subsection{Joint estimating equations}\label{app:estimating-equation-array}

Let $\mathcal A = \{(g,l):g\in\mathcal G_l,\ l\in\mathcal L\}$.
For each $(g,l)\in\mathcal A$ and $Q\in\{TE,SW\}$, the outcome regression system includes
\begin{align*}
  0 & = \sum_{i=1}^N q_{m,Q,i}^{g,l}(\beta_{g,l}^Q),
\end{align*}
using the conditional mean estimating functions described in Appendix~\ref{app:parametric-first-stages}.
The effect equation for the Total effect is
\begin{align*}
  0 & = \sum_{i=1}^N\left[\dfrac{\varpi_i\1\{G_i=g\}}{s_{\varpi,g}}\{Y_{it}-Y_{i,t_b}-m_{\varpi,t}^{TE,g,l}(X_i,H^0_{it};\beta_{g,l}^{TE})\}-\tau_\varpi^{TE}(g,l)\right].
\end{align*}
The effect equation for the Switching effect is
\begin{align*}
  0 & = \sum_{i=1}^N\left[\dfrac{\varpi_i\1\{G_i=g\}}{s_{\varpi,g}}\{Y_{it}-Y_{i,t_b}-m_{\varpi,t}^{SW,g,l}(X_i,H^1_{it};\beta_{g,l}^{SW})\}-\tau_\varpi^{SW}(g,l)\right].
\end{align*}
For the spillover effect among units that never adopt, let $\kappa_{\varpi,\infty}(g,l)$ denote the root of the equation for a given $(g,l)$. Using the conditional mean parameter for the Total effect, its effect equation is
\begin{align*}
  0 & = \sum_{i=1}^N\varpi_i\1\{G_i=\infty\}\{Y_{it}-Y_{i,t_b}-m_{\varpi,t}^{TE,g,l}(X_i,H^0_{it};\beta_{g,l}^{TE})-\kappa_{\varpi,\infty}(g,l)\}.
\end{align*}
Under the weighted identification conditions and correct specification, $\kappa_{\varpi,\infty}(g,l)=\tau_{\varpi,\infty}^{SP}(g,l)$.
The DR system stacks the same conditional mean equations and, for each $Q\in\{TE,SW\}$, the odds equation
\begin{align*}
  0 & = \sum_{i=1}^N q_{\lambda,Q,i}^{g,l}(\alpha_{g,l}^Q),
\end{align*}
together with the following effect equations:
\begin{align*}
  0 & = \sum_{i=1}^N\bigg[ \dfrac{\varpi_i\1\{G_i=g\}}{s_{\varpi,g}}\{Y_{it}-Y_{i,t_b}-m_{\varpi,t}^{TE,g,l}(X_i,H^0_{it};\beta_{g,l}^{TE})\}\\
                      &-\dfrac{\varpi_i\1\{G_i=\infty\}}{s_{\varpi,g}}\lambda_{\varpi,TE}^{g,l}(X_i,H^1_{it};\alpha_{g,l}^{TE})\{Y_{it}-Y_{i,t_b}-m_{\varpi,t}^{TE,g,l}(X_i,H^1_{it};\beta_{g,l}^{TE})\} -\tau_\varpi^{TE}(g,l) \bigg], \\
  0 & = \sum_{i=1}^N\bigg[ \dfrac{\varpi_i\1\{G_i=g\}}{s_{\varpi,g}}\{Y_{it}-Y_{i,t_b}-m_{\varpi,t}^{SW,g,l}(X_i,H^1_{it};\beta_{g,l}^{SW})\}\\
                      &-\dfrac{\varpi_i\1\{G_i=\infty\}}{s_{\varpi,g}}\lambda_{\varpi,SW}^{g,l}(X_i,H^1_{it};\alpha_{g,l}^{SW})\{Y_{it}-Y_{i,t_b}-m_{\varpi,t}^{SW,g,l}(X_i,H^1_{it};\beta_{g,l}^{SW})\} -\tau_\varpi^{SW}(g,l) \bigg].
\end{align*}
The equation for the spillover effect among units that never adopt is the same in the outcome regression and DR systems.
For each $g\in\bigcup_{l\in\mathcal L}\mathcal G_l$, the mass equation is
\begin{align*}
  0 & = \sum_{i=1}^N\{\varpi_i\1\{G_i=g\}-s_{\varpi,g}\}.
\end{align*}
The stacked representation is used to derive the joint linearization and covariance.
Because the nuisance and mass equations do not depend on the effect parameters, the equations may be solved jointly or by solving the nuisance and mass equations before the effect equations.
Let $\theta_{\varpi,N}$ collect the nuisance parameters, effect parameters, and cohort masses in these equations, and let $\Psi_{i,N}^\varpi(\theta)$ be unit $i$'s contribution to the stack.
Define
\begin{align*}
  \widehat\Psi_N^\varpi(\theta) & = \dfrac{1}{N}\sum_{i=1}^N\Psi_{i,N}^\varpi(\theta),                            \\
  \Psi_N^\varpi(\theta)         & = \dfrac{1}{N}\sum_{i=1}^N\mathbb E[\Psi_{i,N}^\varpi(\theta)\mid\mathcal C_N].
\end{align*}
Let $\widehat\theta_{\varpi,N}$ denote the estimator defined by the sample equations, and let $\theta_{\varpi,N}^0$ solve $\Psi_N^\varpi(\theta_{\varpi,N}^0)=0$.
The maps $a_{\varpi,l}^{TE}$, $a_{\varpi,l}^{SW}$, and $a_{\varpi,l}^{\infty,SP}$ apply the cohort weights in Section~\ref{sec:analysis-weights} to the corresponding cell effects and masses at event time $l$.
Define $a_{\varpi,l}^{SP}=a_{\varpi,l}^{TE}-a_{\varpi,l}^{SW}$.
Stack these maps over $l\in\mathcal L$ in $a_\varpi(\theta)$, and define
\begin{align*}
  \boldsymbol\zeta_\varpi & = \left((\tau_\varpi^{TE}(l))_{l\in\mathcal L}^\top,(\tau_\varpi^{SW}(l))_{l\in\mathcal L}^\top,(\tau_\varpi^{SP}(l))_{l\in\mathcal L}^\top,(\bar\tau_{\varpi,\infty}^{SP}(l))_{l\in\mathcal L}^\top\right)^\top.
\end{align*}
Thus $\theta_{\varpi,N}$ retains the nuisance parameters used in the joint linearization, whereas $a_\varpi$ maps the full parameter vector to the event-time effect quantities collected in $\boldsymbol\zeta_\varpi$.
The system is exactly identified and has fixed dimension when $\mathcal L$, $\mathcal A$, and the nuisance specifications are fixed.

\section{Inference details and proofs}\label{app:inference-proofs}

This appendix states the conditions for the finite-dimensional estimating equations and proves the results in Section~\ref{sec:inference}.
All weak-convergence statements below are conditional on $\mathcal C_N$.
Conditional weak convergence means that the bounded-Lipschitz distance between the conditional law of a statistic and its stated limiting law converges in probability to zero.
Set $\varpi_i=1$ in Appendix~\ref{app:estimating-equation-array} and denote the resulting score array, population root, aggregation map, and target vector by $\Psi_{i,N}$, $\theta_N^0$, $a$, and $\boldsymbol\zeta$.

\subsection{Finite-dimensional regularity}\label{app:inference-regularity}

Define the population Jacobian, aggregation derivative, and long-run score covariance by
\begin{align*}
  J_N & = \nabla_\theta\Psi_N(\theta_N^0),\   A_N = \left.\dfrac{\partial a(\theta)}{\partial\theta^\top}\right|_{\theta=\theta_N^0},\ \Omega_N = \dfrac{1}{N}\operatorname{Var}\!\left\{\sum_{i=1}^N\Psi_{i,N}(\theta_N^0)\mathrel{\Big|}\mathcal C_N\right\}.
\end{align*}
The influence contribution for the stacked target and its covariance limit are
\begin{align*}
  \varphi_{i,N} & = -A_NJ_N^{-1}\Psi_{i,N}(\theta_N^0),\ \Gamma = AJ^{-1}\Omega J^{-\top}A^\top.
\end{align*}
For $Q\in\{\mathrm{TE},\mathrm{SW},\mathrm{SP}\}$, let $\varphi_{i,N}^Q(l)$ denote the row of $\varphi_{i,N}$ corresponding to effect $Q$ at event time $l$.
Let $\varphi_{i,N}^{\infty,SP}(l)$ denote the row for the spillover effect among units that never adopt.
The row for the spillover effect for the target cohorts satisfies $\varphi_{i,N}^{SP}(l)=\varphi_{i,N}^{TE}(l)-\varphi_{i,N}^{SW}(l)$.
The covariance notation in Section~\ref{sec:inference} refers to the corresponding entries of $\Gamma$.

For feasible covariance estimation, define
\begin{align*}
  \widehat J_N & = \nabla_\theta\widehat\Psi_N(\widehat\theta_N), \ \widehat A_N = \left.\dfrac{\partial a(\theta)}{\partial\theta^\top}\right|_{\theta=\widehat\theta_N},\ \widehat\varphi_{i,N} = -\widehat A_N\widehat J_N^{-1}\Psi_{i,N}(\widehat\theta_N).
\end{align*}
The coordinate $\widehat\varphi_{i,N}^Q(l)$ is the corresponding row of $\widehat\varphi_{i,N}$, and $\widehat\varphi_{i,N}^{\infty,SP}(l)$ is its row for the spillover effect among units that never adopt.
The feasible contribution for the spillover effect for the target cohorts satisfies $\widehat\varphi_{i,N}^{SP}(l)=\widehat\varphi_{i,N}^{TE}(l)-\widehat\varphi_{i,N}^{SW}(l)$.

\begin{assumption}[GMM and dependence regularity]\label{assumption:inf_regular}
  There exist constants \(c_0 > 0\) and \(C_0 < \infty\), independent of \(N\), such that:
  \begin{enumerate}
    \item[(i)] The stacked equations and $\theta_N$ have the same fixed dimension, and the parameter space $\Theta$ is compact and convex.
          A measurable $\widehat\theta_N\in\Theta$ exists and satisfies
          \begin{align*}
            \|\widehat\Psi_N(\widehat\theta_N)\| & = o_p(N^{-1/2}).
          \end{align*}
    \item[(ii)] The solution $\theta_N^0$ is in the interior of $\Theta$ and, for every fixed $\epsilon > 0$,
          \begin{align*}
            \Pr (\inf_{\theta\in\Theta:\|\theta - \theta_N^0\|\ge\epsilon}\|\Psi_N(\theta)\|\ge c_\epsilon) & \to 1
          \end{align*}
          for some \(c_\epsilon > 0\).
    \item[(iii)] The equations are continuously differentiable on an open neighborhood of $\Theta$. There are random envelopes $B_{i,N}$ and $L_{i,N}$ such that
          \begin{align*}
            \sup_{\theta\in\Theta} \{\|\Psi_{i,N}(\theta)\| + \|\nabla_\theta\Psi_{i,N}(\theta)\|\}                                                 & \le B_{i,N},                              \\
            \|\Psi_{i,N}(\theta) - \Psi_{i,N}(\widetilde\theta)\| + \|\nabla_\theta\Psi_{i,N}(\theta) - \nabla_\theta\Psi_{i,N}(\widetilde\theta)\| & \le L_{i,N}\|\theta - \widetilde\theta\|,
          \end{align*}
          and, for some \(p > 4\),
          \begin{align*}
            \sup_{i,N} \mathbb E[B_{i,N}^p + L_{i,N}^p\mid\mathcal C_N] & \le C_0.
          \end{align*}
    \item[(iv)] The equations and their derivative arrays satisfy
          \begin{align*}
            \sup_{\theta\in\Theta} \|\widehat\Psi_N(\theta) - \Psi_N(\theta)\|                           & = o_p(1), \\
            \sup_{\theta\in\Theta} \|\nabla_\theta\widehat\Psi_N(\theta) - \nabla_\theta\Psi_N(\theta)\| & = o_p(1).
          \end{align*}
    \item[(v)] The population Jacobians are uniformly nonsingular and converge to a deterministic nonsingular matrix $J$: $ c_0 \le \sigma_{\min}(J_N) \le \sigma_{\max}(J_N) \le C_0$ and $J_N \xrightarrow{p}J$.
    \item[(vi)] For every $Q\in\{\mathrm{TE},\mathrm{SW},\mathrm{SP}\}$ and $l\in\mathcal L$, the maps $a_l^Q(\theta)$ and $a_l^{\infty,SP}(\theta)$ are continuously differentiable near $\theta_N^0$, with uniformly bounded and Lipschitz gradients. The stacked derivative satisfies $A_N \xrightarrow{p}A$ for a deterministic finite matrix $A$.
    \item[(vii)] The long-run score covariance satisfies $\Omega_N \xrightarrow{p}\Omega$ for a deterministic finite positive-semidefinite matrix $\Omega$.
  \end{enumerate}
\end{assumption}

Parts (i)--(iv) give fixed-dimensional identification, smoothness, moment bounds, and uniform convergence of the sample equations and their derivatives for the exactly identified system \citep{newey1994large}.
Parts (v)--(vii) give the deterministic Jacobian, aggregation, and covariance limits.

\begin{remark}\label{app:cohort_specific_inference}
  Fix a cohort and event-time pair $(g,l)$ and replace the corresponding aggregation row of $a(\theta)$ by the coordinate for the effect at $(g,l)$.
  Suppose that this population coordinate equals the cohort-specific target.
  If Assumptions~\ref{assumption:inf_support}, \ref{assumption:inf_regular}, \ref{assumption:inf_spatial}, \ref{assumption:inf_shac}, and \ref{assumption:inf_shac_technical} hold for its influence contribution, the results in Section~\ref{sec:inference} apply to that cell.
  The number of selected cohort and event-time pairs remains fixed.
\end{remark}

\begin{remark}[Joint inference for pre-tests]\label{app:pretest_inference}
  For any fixed collection of pre-treatment periods, augment the aggregation map with the rows corresponding to the selected pre-treatment parameters.
  Under the corresponding conditions, the same linearization, joint normality, and spatial HAC conclusions apply to their estimators.
  If the conditional means of their influence contributions vanish and the limiting covariance matrix is positive definite, the usual Wald statistic for the joint null that these parameters are zero converges to a chi-square distribution with degrees of freedom equal to the number of selected parameters.
\end{remark}

\subsection{Spatial CLT rates}\label{app:spatial-clt-rates}

For any real $r$, define
\begin{align*}
  \mathcal N_N(i;r) & = \begin{cases} \{j\in\{1,\ldots,N\}:\rho(i,j)\le r\}, & r\ge 0, \\
                                      \varnothing,                           & r<0.
                        \end{cases}
\end{align*}
For a fixed scalar direction $v$ of the score vector, define
\begin{align*}
  Z_{i,N}(v)    & = v^\top\!\left[\Psi_{i,N}(\theta_N^0)-\mathbb E[\Psi_{i,N}(\theta_N^0)\mid\mathcal C_N]\right], \\
  \sigma_N^2(v) & = \operatorname{Var}\!\left\{\sum_{i=1}^N Z_{i,N}(v)\mathrel{\Big|}\mathcal C_N\right\}.
\end{align*}

\begin{assumption}[Spatial CLT rates]\label{assumption:inf_spatial_rates}
  There is a sequence of positive integers $r_N\to\infty$ such that, for every fixed $v$ for which $\sigma_N^2(v)/N$ converges in probability to a positive deterministic limit, the following conditions hold.
  \begin{enumerate}
    \item[(i)] For $k=1,2$,
          \begin{align*}
            \sum_{i=1}^N\sum_{s=0}^{\infty}\sum_{\substack{1\le j\le N\\s\le\rho(i,j)<s+1}}\dfrac{\left|\mathcal N_N(i;r_N)\setminus\mathcal N_N(j;s-1)\right|^k\eta_N(s)^{1-(2+k)/p}}{\sigma_N(v)^{2+k}} & \xrightarrow{p}0.
          \end{align*}
    \item[(ii)] $N^2\eta_N(r_N)^{1-1/p}/\sigma_N(v)\xrightarrow{p}0$.
  \end{enumerate}
\end{assumption}

The sequence $r_N$ is a theoretical truncation radius for the scalar CLT and is distinct from the feasible spatial HAC bandwidth $b_N$.

\subsection{Spatial HAC conditions}\label{app:shac-technical}

Let
\begin{align*}
  \mu_{i,N} & = \mathbb E[\varphi_{i,N}\mid\mathcal C_N],\ u_{i,N} = \varphi_{i,N}-\mu_{i,N},\ \Gamma_N^E = \dfrac{1}{N}\sum_{i=1}^N\sum_{j=1}^N K\!\left(\dfrac{\rho(i,j)}{b_N}\right)\mu_{i,N}\mu_{j,N}^\top.
\end{align*}
\begin{assumption}[Spatial HAC conditions]\label{assumption:inf_shac_technical}
  The following conditions hold.
  \begin{enumerate}
    \item[(i)] The rate condition for covariance estimation is
          \begin{align*}
            \dfrac{1}{N^2}\sum_{i=1}^N\sum_{s=0}^{\infty}\sum_{\substack{1\le j\le N\\s\le\rho(i,j)<s+1}}\left|\mathcal N_N(i;b_N)\setminus\mathcal N_N(j;s-1)\right|^2\eta_N(s)^{1-4/p} & \xrightarrow{p}0.
          \end{align*}
    \item[(ii)] The feasible influence contributions satisfy
          \begin{align*}
            \left\|\dfrac{1}{N}\sum_{i=1}^N\sum_{j=1}^N K\!\left(\dfrac{\rho(i,j)}{b_N}\right)\left\{\widehat\varphi_{i,N}\widehat\varphi_{j,N}^\top-\varphi_{i,N}\varphi_{j,N}^\top\right\}\right\| & = o_p(1).
          \end{align*}
    \item[(iii)] The cross terms involving the centered contribution and its conditional mean satisfy
          \begin{align*}
            \dfrac{1}{N}\sum_{i=1}^N\sum_{j=1}^N K\!\left(\dfrac{\rho(i,j)}{b_N}\right)u_{i,N}\mu_{j,N}^\top & = o_p(1), \\
            \dfrac{1}{N}\sum_{i=1}^N\sum_{j=1}^N K\!\left(\dfrac{\rho(i,j)}{b_N}\right)\mu_{i,N}u_{j,N}^\top & = o_p(1).
          \end{align*}
    \item[(iv)] The centering matrix satisfies $\Gamma_N^E \xrightarrow{p}\Gamma^E$ for a finite deterministic matrix $\Gamma^E$.
  \end{enumerate}
\end{assumption}

Define the full uncentered SHAC target by $\Gamma^+ = \Gamma+\Gamma^E$.
For $Q,Q'\in\{\mathrm{TE},\mathrm{SW},\mathrm{SP}\}$, let $\Gamma_{QQ'}^E(l,l')$ denote the corresponding element of $\Gamma^E$.
The centering covariance for the spillover effect for the target cohorts under no own adoption is
\begin{align*}
  \Gamma_{SP,SP}^E(l,l') & = \Gamma_{TE,TE}^E(l,l')+\Gamma_{SW,SW}^E(l,l')-\Gamma_{TE,SW}^E(l,l')-\Gamma_{SW,TE}^E(l,l').
\end{align*}
Let $\Gamma_{\infty}^{E,SP}(l,l')$ denote the element of $\Gamma^E$ corresponding to two rows for spillover effects among units that never adopt.
Parts (ii) and (iii) state the requirements for replacement by feasible influence rows and for the cross terms used by the uncentered SHAC estimator.
They are technical conditions for covariance estimation and do not restrict the identifying comparisons.

\subsection{Weighted inference}\label{app:weighted-inference}

Define the weighted analogues of the objects in Appendix~\ref{app:inference-regularity} using $\widehat\Psi_N^\varpi$, $\Psi_N^\varpi$, $\widehat\theta_{\varpi,N}$, $\theta_{\varpi,N}^0$, and $a_\varpi$.
Denote the population Jacobian, aggregation derivative, and long-run score covariance by $J_{\varpi,N}$, $A_{\varpi,N}$, and $\Omega_{\varpi,N}$.
Denote the population and feasible influence contributions by $\varphi_{i,N}^\varpi$ and $\widehat\varphi_{i,N}^\varpi$.

\begin{assumption}[Weighted score array]\label{assumption:weighted-inference}
  There exist constants $c_0>0$ and $C_0<\infty$, independent of $N$, such that the following conditions hold.
  \begin{enumerate}
    \item[(i)] The identifying conditions in Section~\ref{sec:identification} hold after reweighting the relevant populations by $\varpi_i$, and $a_\varpi(\theta_{\varpi,N}^0)=\boldsymbol\zeta_\varpi$. When DR estimation is used, $0\le\lambda_{\varpi,TE}^{g,l}(x,h)\le C_0$ and $0\le\lambda_{\varpi,SW}^{g,l}(x,h)\le C_0$ on the corresponding supports among units that never adopt.
    \item[(ii)] The weighted masses satisfy $\inf_{g\in\bigcup_{l\in\mathcal L}\mathcal G_l}s_{\varpi,g} \ge c_0$ and $s_{\varpi,\infty}\ge c_0$.
          For every $l\in\mathcal L$, $g\in\mathcal G_l$, and $h\in\mathcal H_N$, whenever
          \begin{align*}
            \mathbb E[\varpi_i\1\{G_i = g,H^0_{i,g + l} = h\}] + \mathbb E[\varpi_i\1\{G_i = g,H^1_{i,g + l} = h\}] + \mathbb E[\varpi_i\1\{G_i = \infty,H^0_{i,g + l} = h\}] & > 0.
          \end{align*}
          the corresponding weighted mass among units that never adopt satisfies
          \begin{align*}
            \mathbb E[\varpi_i\1\{G_i = \infty,H^1_{i,g + l} = h\}] & \ge c_0.
          \end{align*}
    \item[(iii)] Assumption~\ref{assumption:inf_regular} holds for the weighted system. The deterministic limits of $J_{\varpi,N}$, $A_{\varpi,N}$, and $\Omega_{\varpi,N}$ are denoted by $J_\varpi$, $A_\varpi$, and $\Omega_\varpi$.
    \item[(iv)] Assumption~\ref{assumption:inf_spatial} holds for the centered weighted score array. In Assumption~\ref{assumption:inf_spatial_rates}, $Z_{i,N}(v)$ and $\sigma_N^2(v)$ are formed from $\Psi_{i,N}^\varpi(\theta_{\varpi,N}^0)$, and the rate conditions hold with its conditional dependence coefficients.
    \item[(v)] Assumptions~\ref{assumption:inf_shac} and \ref{assumption:inf_shac_technical} hold for $\varphi_{i,N}^\varpi$ and $\widehat\varphi_{i,N}^\varpi$. The conditions for feasible replacement, the cross terms, and convergence of the centering matrix are imposed on these weighted contributions.
  \end{enumerate}
\end{assumption}

Under Assumption~\ref{assumption:weighted-inference}, the joint asymptotic normality, spatial HAC conclusions, and pointwise coverage statements in Section~\ref{sec:inference} hold for the weighted estimators.

\subsection{Proofs}

\begin{proof}[Proof of Theorem~\ref{thm:event_time_distribution}]
  Assumption~\ref{assumption:inf_support} makes the effect equations and aggregation maps well defined.
  For every fixed $\epsilon>0$, Assumption~\ref{assumption:inf_regular}(ii) and (iv) give, with probability approaching one,
  \begin{align*}
    \inf_{\theta\in\Theta:\|\theta-\theta_N^0\|\ge\epsilon}\|\widehat\Psi_N(\theta)\| & \ge \inf_{\theta\in\Theta:\|\theta-\theta_N^0\|\ge\epsilon}\|\Psi_N(\theta)\| -\sup_{\theta\in\Theta}\|\widehat\Psi_N(\theta)-\Psi_N(\theta)\| \\
                                                                                      & \ge c_\epsilon-o_p(1).
  \end{align*}
  The approximate-root condition therefore implies $\widehat\theta_N\xrightarrow{p}\theta_N^0$.
  Since $\Psi_N(\theta_N^0)=0$, Assumption~\ref{assumption:inf_regular}(vii) gives
  \begin{align*}
    \mathbb E\!\left[ \left\|\dfrac{1}{\sqrt N}\sum_{i=1}^N\Psi_{i,N}(\theta_N^0)\right\|^2 \mathrel{\Big|}\mathcal C_N \right] & = \operatorname{tr}(\Omega_N)=O_p(1).
  \end{align*}
  Hence $N^{-1/2}\sum_{i=1}^N\Psi_{i,N}(\theta_N^0)=O_p(1)$.
  A mean-value expansion of the estimating equations gives
  \begin{align*}
    \widehat\Psi_N(\widehat\theta_N) & = \widehat\Psi_N(\theta_N^0) +\left\{\int_0^1\nabla_\theta\widehat\Psi_N\!\left(\theta_N^0+s(\widehat\theta_N-\theta_N^0)\right)\,ds\right\} (\widehat\theta_N-\theta_N^0).
  \end{align*}
  Assumption~\ref{assumption:inf_regular}(iii)--(v) implies
  \begin{align*}
    \int_0^1\nabla_\theta\widehat\Psi_N\!\left(\theta_N^0+s(\widehat\theta_N-\theta_N^0)\right)\,ds & = J_N+o_p(1).
  \end{align*}
  Since $\widehat\Psi_N(\theta_N^0)=N^{-1}\sum_{i=1}^N\Psi_{i,N}(\theta_N^0)$, the preceding two displays and the approximate-root condition yield
  \begin{align*}
    \sqrt N(\widehat\theta_N-\theta_N^0) & = -\{J_N+o_p(1)\}^{-1}\dfrac{1}{\sqrt N}\sum_{i=1}^N\Psi_{i,N}(\theta_N^0)+o_p(1) \\
                                         & = -J_N^{-1}\dfrac{1}{\sqrt N}\sum_{i=1}^N\Psi_{i,N}(\theta_N^0)+o_p(1).
  \end{align*}
  It follows that $\widehat\theta_N-\theta_N^0=O_p(N^{-1/2})$.
  The target premise gives the corresponding equalities for the three primitive rows.
  The identity $\tau^{SP}=\tau^{TE}-\tau^{SW}$ gives the equality for the Spillover row, so $a(\theta_N^0)=\boldsymbol\zeta$.
  Assumption~\ref{assumption:inf_regular}(vi) gives
  \begin{align*}
    \sqrt N\left\|a(\widehat\theta_N)-a(\theta_N^0)-A_N(\widehat\theta_N-\theta_N^0)\right\| & = O_p\!\left(\sqrt N\|\widehat\theta_N-\theta_N^0\|^2\right)=o_p(1).
  \end{align*}
  Therefore,
  \begin{align*}
    \sqrt N\{a(\widehat\theta_N)-a(\theta_N^0)\} & = A_N\sqrt N(\widehat\theta_N-\theta_N^0)+o_p(1)                          \\
                                                 & = -A_NJ_N^{-1}\dfrac{1}{\sqrt N}\sum_{i=1}^N\Psi_{i,N}(\theta_N^0)+o_p(1) \\
                                                 & = \dfrac{1}{\sqrt N}\sum_{i=1}^N\varphi_{i,N}+o_p(1).
  \end{align*}
\end{proof}

\begin{proof}[Proof of Corollary~\ref{cor:joint_asymptotic_normality}]
  The conditional population equation gives
  \begin{align*}
    \sum_{i = 1}^N \mathbb E[\Psi_{i,N}(\theta_N^0)\mid\mathcal C_N] & = 0,
  \end{align*}
  so the sum of the scores equals the sum of the centered scores.
  Fix a scalar direction $c$ of the stacked effect vector and set $v_N = -J_N^{-\top}A_N^\top c$ and $v = -J^{-\top}A^\top c$.
  Assumption~\ref{assumption:inf_regular}(v) and (vi) give $v_N\xrightarrow{p}v$.
  Assumption~\ref{assumption:inf_regular}(vii) gives
  \begin{align*}
    (v_N-v)^\top\dfrac{1}{\sqrt N}\sum_{i=1}^N \left[\Psi_{i,N}(\theta_N^0)-\mathbb E[\Psi_{i,N}(\theta_N^0)\mid\mathcal C_N]\right] & = o_p(1).
  \end{align*}
  Let $R_N$ denote either the $o_p(1)$ remainder in Theorem~\ref{thm:event_time_distribution} or the replacement term in the preceding display.
  For every $\epsilon>0$,
  \begin{align*}
    \mathbb E\!\left[\Pr (\|R_N\|>\epsilon\mid\mathcal C_N)\right] & = \Pr (\|R_N\|>\epsilon)\to0,
  \end{align*}
  and Markov's inequality gives $\Pr (\|R_N\|>\epsilon\mid\mathcal C_N)\xrightarrow{p}0$.
  Thus both remainders vanish in conditional probability.
  If $v^\top\Omega v>0$, then $\sigma_N^2(v)/N=v^\top\Omega_Nv\xrightarrow{p}v^\top\Omega v>0$.
  By the subsequence principle, every subsequence contains a further subsequence along which the rate conditions in Assumption~\ref{assumption:inf_spatial_rates} and the variance convergence in the preceding display hold almost surely.
  The spatial adaptation in the proof of Theorem~4.2 in \citet*{XU2026106304}, which applies Theorem~3.2 in \citet*{KOJEVNIKOV2021882}, then gives
  \begin{align*}
    \dfrac{1}{\sigma_N(v)}\sum_{i=1}^N Z_{i,N}(v) & \xrightarrow{d}\mathcal N(0,1)
  \end{align*}
  conditionally on $\mathcal C_N$.
  Conditional Slutsky's theorem therefore gives
  \begin{align*}
    c^\top\dfrac{1}{\sqrt N}\sum_{i=1}^N\varphi_{i,N} & \xrightarrow{d}\mathcal N(0,c^\top\Gamma c)
  \end{align*}
  conditionally on $\mathcal C_N$.
  If $v^\top\Omega v=0$, covariance convergence and conditional Chebyshev's inequality instead make this scalar leading term vanish in conditional probability.
  The conditional Cram\'er--Wold device gives the stated joint Gaussian limit, including a singular $\Gamma$.
\end{proof}

\begin{proof}[Proof of Proposition~\ref{prop:shac_covariance}]
  Write $K_{ij,N}=K\{\rho(i,j)/b_N\}$ within this proof.
  Assumptions~\ref{assumption:inf_spatial}(ii), \ref{assumption:inf_shac}(i)--(iii), and \ref{assumption:inf_shac_technical}(i) give the kernel approximation and bounds on the variance of the quadratic form for the centered population influence contributions.
  Applying the same subsequence argument to these in-probability conditions, the spatial adaptation in the proof of Theorem~4.3 in \citet*{XU2026106304}, which applies Proposition~4.1 in \citet*{KOJEVNIKOV2021882}, gives
  \begin{align*}
    \dfrac{1}{N}\sum_{i=1}^N\sum_{j=1}^N K_{ij,N}u_{i,N}u_{j,N}^\top & = \Gamma+o_p(1).
  \end{align*}
  Assumption~\ref{assumption:inf_shac_technical}(ii) gives
  \begin{align*}
    \dfrac{1}{N}\sum_{i=1}^N\sum_{j=1}^N K_{ij,N} \widehat\varphi_{i,N}\widehat\varphi_{j,N}^\top & = \dfrac{1}{N}\sum_{i=1}^N\sum_{j=1}^N K_{ij,N} \varphi_{i,N}\varphi_{j,N}^\top+o_p(1).
  \end{align*}
  Using $\varphi_{i,N}=u_{i,N}+\mu_{i,N}$, expand the population quadratic form as
  \begin{align*}
    \dfrac{1}{N}\sum_{i=1}^N\sum_{j=1}^N K_{ij,N}\varphi_{i,N}\varphi_{j,N}^\top & = \dfrac{1}{N}\sum_{i=1}^N\sum_{j=1}^N K_{ij,N}u_{i,N}u_{j,N}^\top+\dfrac{1}{N}\sum_{i=1}^N\sum_{j=1}^N K_{ij,N}u_{i,N}\mu_{j,N}^\top \\
                                                                                 & +\dfrac{1}{N}\sum_{i=1}^N\sum_{j=1}^N K_{ij,N}\mu_{i,N}u_{j,N}^\top +\Gamma_N^E                                                       \\
                                                                                 & = \Gamma+\Gamma^E+o_p(1).
  \end{align*}
  The last equality follows from Assumption~\ref{assumption:inf_shac_technical}(iii) and (iv).
  Thus $\|\widehat\Gamma-\Gamma^+\|=o_p(1)$.
  On the event that the kernel weight matrix is positive semidefinite, every fixed vector $q$ satisfies
  \begin{align*}
    q^\top\Gamma_N^Eq & = \dfrac{1}{N}\sum_{i=1}^N\sum_{j=1}^N K_{ij,N} (q^\top\mu_{i,N})(q^\top\mu_{j,N})\ge0.
  \end{align*}
  This event has probability approaching one.
  Hence $\Gamma^E$ is positive semidefinite, which gives the conservative diagonal statement.
  If all conditional influence means vanish, $\Gamma^E=0$, which gives consistency of the covariance estimator for every component and pair of components.
  Applying the contrast between the Total and Switching effects gives the centering covariance $\Gamma_{SP,SP}^E$ for the spillover effect for the target cohorts under no own adoption defined above.
\end{proof}

\begin{proof}[Proof of Corollary~\ref{cor:pointwise_ci}]
  Corollary~\ref{cor:joint_asymptotic_normality} gives the Gaussian limit for the jointly estimated effects.
  The conditional-probability argument in its proof applies to the covariance remainder in Proposition~\ref{prop:shac_covariance}.
  For $Q\in\{\mathrm{TE},\mathrm{SW},\mathrm{SP}\}$, Proposition~\ref{prop:shac_covariance} and Slutsky's theorem give
  \begin{align*}
    \dfrac{\sqrt N\{\widehat\tau^Q(l)-\tau^Q(l)\}}{\sqrt{\widehat\Gamma_{QQ}(l,l)}} & \xrightarrow{d} \mathcal N\!\left(0,\dfrac{\Gamma_{QQ}(l,l)}{\Gamma_{QQ}(l,l)+\Gamma_{QQ}^E(l,l)}\right)
  \end{align*}
  conditionally on $\mathcal C_N$.
  Therefore,
  \begin{align*}
    \Pr (\tau^Q(l)\in\widehat C_Q^{pt}(l)\mid\mathcal C_N) & \xrightarrow{p} 2\Phi\!\left(z_{1-\alpha/2}\sqrt{\dfrac{\Gamma_{QQ}(l,l)+\Gamma_{QQ}^E(l,l)}{\Gamma_{QQ}(l,l)}}\right)-1.
  \end{align*}
  This limit is at least $1-\alpha$.
  For $Q\in\{\mathrm{TE},\mathrm{SW}\}$, $\Gamma_{QQ}^E(l,l)$ is the corresponding diagonal element of $\Gamma^E$.
  For $Q=\mathrm{SP}$, it is the contrast $\Gamma_{SP,SP}^E(l,l)$ between the Total and Switching effects defined above.
  The limit equals $1-\alpha$ when $\mathbb E[\varphi_{i,N}^Q(l)\mid\mathcal C_N]=0$ for every $i$.
  The same calculation applies to the coordinate for the spillover effect among units that never adopt, with variance $\Gamma_{\infty}^{SP}(l,l)$ and centering term $\Gamma_{\infty}^{E,SP}(l,l)$.
\end{proof}

\section{Sensitivity to the exposure mapping}\label{app:application-mapping-sensitivity}

We replace the 75-mile radius in Equation~\eqref{eq:application_exposure} with 25, 50, and 100 miles.
The number of other counties within 50 miles that have adopted by date $t$ is grouped as zero, one, or at least two.
The number of years since exposure within 50 miles first becomes positive is grouped as zero, one to two years, or at least three years.
The distance to the nearest other county that has adopted by date $t$ defines four exposure levels: no adopter within 100 miles, an adopter within 50 miles, an adopter between 50 and 75 miles, and an adopter between 75 and 100 miles.

\begin{figure}[h]
  \centering
  \begin{subfigure}[t]{0.32\textwidth}
    \centering
    \includegraphics[width=\linewidth]{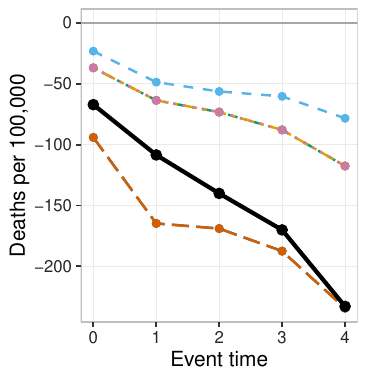}
    \caption{Total effect}
    \label{fig:application_mapping_sensitivity_total_v2}
  \end{subfigure}
  \hfill
  \begin{subfigure}[t]{0.32\textwidth}
    \centering
    \includegraphics[width=\linewidth]{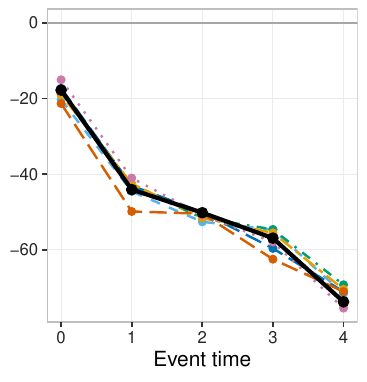}
    \caption{Switching effect}
    \label{fig:application_mapping_sensitivity_switching_v2}
  \end{subfigure}
  \hfill
  \begin{subfigure}[t]{0.32\textwidth}
    \centering
    \includegraphics[width=\linewidth]{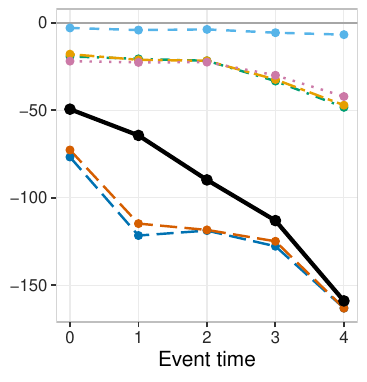}
    \caption{Spillover effect}
    \label{fig:application_mapping_sensitivity_spillover_v2}
  \end{subfigure}
  \par\smallskip
  \includegraphics[width=0.96\textwidth]{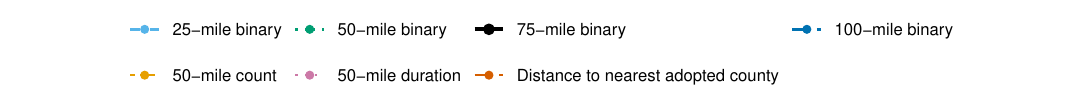}
  \caption{Estimated mortality effects under alternative exposure mappings}
  \label{fig:application_mapping_sensitivity_v2}
  \begin{minipage}{0.94\linewidth}
    \footnotesize Notes: The lines report point estimates at event times $l = 0$ through $l = 4$. The black line corresponds to the exposure mapping with a 75-mile radius.
  \end{minipage}
\end{figure}

Figure~\ref{fig:application_mapping_sensitivity_v2} reports the Total, Switching, and Spillover estimates at event times $l = 0$ through $l = 4$.
All three estimates are negative at every event time under each mapping.
The Total and Switching estimates become more negative over event time under all seven mappings.
The Spillover estimates generally become more negative, although the path is not monotone under every mapping.
The estimates based on the count and duration mappings within 50 miles are close to those from the binary mapping with a radius of 50 miles, while the mappings with radii of 75 and 100 miles produce larger Total and Spillover estimates in absolute value.
Each mapping defines different exposure levels and therefore different total, switching, and spillover effects.
\end{document}